\documentclass[10pt,journal,cspaper,compsoc]{IEEEtran}
 \usepackage[cmex10]{amsmath}
 \usepackage{amsfonts}
\usepackage{url}
\usepackage{gnuplottex,color}

\usepackage{times}
\usepackage{amssymb,amsmath}
\usepackage{amsmath}
\usepackage{graphicx}
 \usepackage{ifpdf}
 \usepackage{amsthm}
\usepackage{algorithm}
\usepackage[algo2e, ruled, vlined]{algorithm2e}

\newcommand{\ignore}[1]{}
\newcommand{\notinproc}[1]{#1}
\newcommand{\onlyinproc}[1]{}

\newtheorem{thm}{Theorem}[section]
\newtheorem{theorem}{Theorem}[section]

\newtheorem{lemma}[thm]{Lemma}

\def\E{{\textsf E}}
\def\var{{\textsf var}}
\def\vecv{\boldsymbol{v}}
\def\vecm{\boldsymbol{m}}
\def\veco{\boldsymbol{o}}
\def\vecf{\boldsymbol{f}}
\def\vecpsi{\boldsymbol{\psi}}
\def\vecphi{\boldsymbol{\phi}}
\def\vecbeta{\boldsymbol{\beta}}
\def\Cap{\textsf{cap}}

\def\Exp{\textsf{Exp}}
\def\Erlang{\textsf{Erlang}}

\begin{document}
\title{Stream Sampling for Frequency Cap Statistics}

\author{Edith Cohen \\ {\tt edith@cohenwang.com} \\ Google Research,
  CA, USA \\Tel Aviv
  University, Israel}
\ignore{
 \numberofauthors{1}
\author{
\alignauthor Edith Cohen\\
       \affaddr{Google Research, CA, USA}\\ Tel Aviv University, Israel \\  
       \email{edith@cohenwang.com}
}}

\ignore{
\newfont{\mycrnotice}{ptmr8t at 7pt}
\newfont{\myconfname}{ptmri8t at 7pt}
\let\crnotice\mycrnotice%
\let\confname\myconfname%

\permission{
Permission to make digital or hard copies of part or all of this work for personal or classroom use is granted without fee provided that copies are not made or distributed for profit or commercial advantage and that copies bear this notice and the full citation on the first page. Copyrights for third-party components of this work must be honored. For all other uses, contact the Owner/Author.}
\conferenceinfo{KDD'15,}{August 10-13, 2015, Sydney, NSW,
  Australia.} 
\copyrightetc{ACM \the\acmcopyr}
\crdata{978-1-4503-3664-2/15/08.\\
http://dx.doi.org/10.1145/2783258.2783279}

\clubpenalty=10000 
\widowpenalty = 10000
}


\SetKwFunction{Hash}{Hash}
\SetKwFunction{seed}{seed}
\SetKwFunction{keybase}{KeyBase}
\SetKwFunction{key}{Key}

\IEEEcompsoctitleabstractindextext{
 \begin{abstract}
Unaggregated data, in a streamed or distributed form, is prevalent and comes from diverse sources such as interactions of users with web services and IP traffic. Data elements have {\em keys}  (cookies, users, queries) and elements with different keys interleave.

  Analytics on such data typically utilizes statistics expressed as
  a sum over keys in a specified segment of a function $f$ applied to the 
frequency (the total number of occurrences) of the key. In particular,
 {\em Distinct} is the number of active keys in the segment,
{\em Sum} is the sum of their frequencies, and both are special cases of
{\em frequency cap}
  statistics, which cap the frequency by a parameter $T$. An
  important application of cap statistics is staging advertisement
  campaigns, where the cap parameter is the maximum
  number of impressions per user and the statistics is the total number of qualifying impressions.

The number of
 distinct active keys in the data can be very large, making exact
 computation of queries costly.  Instead, we can estimate
these statistics from a sample.   An optimal sample for
 a given function $f$ would include a key with frequency $w$ with
 probability roughly proportional to $f(w)$.  But while such  a "gold-standard"
 sample can be easily computed over the aggregated data (the set of
 key-frequency pairs), exact aggregation itself is costly, requiring state
 proportional to the number of active keys.  Ideally,
 we would like to compute a sample without exact aggregation.

We present a sampling framework for unaggregated data that uses a
single pass (for streams) or  two passes (for distributed data) and
state proportional to the desired sample size. Our design unifies
classic solutions for Distinct and Sum. Specifically, our
$\ell$-capped samples provide nonnegative unbiased estimates of any
monotone non-decreasing frequency statistics, and close to
gold-standard estimates for frequency cap statistics with
$T=\Theta(\ell)$.  Furthermore, our design facilitates multi-objective
samples, which provide tight estimates for a specified set of
statistics using a single smaller sample.
 \end{abstract}
}

\maketitle

\SetKwArray{Counters}{Counters}

\section{Introduction}

 The data available from many services, such as interactions
 of users with Web services or content, search logs, and IP traffic,
 is presented in an
{\em unaggregated} form. In this model, each data
{\em element} has a {\em key}  from a
 universe ${\cal X}$ and a weight $w > 0$.
Data elements with different keys interleave in a data stream or
distributed storage.

  The {\em aggregated view} of the data is a set of pairs that consists of an {\em active}
  keys $x\in {\cal X}$  (a key that occurred at least once) and the respective 
  total weight $w_x$  of all elements 
  with key $x$. 
  When all element weights are uniform, $w_x$ is 
the number of occurrences of an element with key $x$ in the data.
The weight $w_x$ is often referred to as the {\em frequency} of the key 
(it is proportional to the actual frequency in the data set).

{\em Frequency statistics} of such data are fundamental  to data analytics.
Queries have the form
\begin{equation}  \label{basicquery:eq}
Q(f,H)  \equiv \sum_{x\in {\cal X} \cap H} f(w_x)\ ,
\end{equation}
where $f(w)\geq 0$ is a nonnegative function such that $f(0) = 0$ and
 $H$ is a
  selection predicate that specifies a segment of the key population
  ${\cal X}$.  
Typically $f$ is monotone non decreasing, which means
  that more frequent keys carry at least the same contribution as less
  frequent ones.   Some prominent examples are the {\em
    $p^{\text{th}}$ frequency
    moment}, where $f(x) = x^p$  for $p>0$ \cite{ams99} and
 {\em frequency cap} statistics, where
$f$ is a {\em cap}  function with parameter $T > 0$:
 $$\Cap_T(y)  \equiv \min \{ y,T \}\ .$$  

  Two special cases of both cap statistics and frequency moments, that
are widely studied and applied in big data analytics,  are  {\em
  Distinct} -- the number of distinct (active) keys in the segment ($L_0$
moment or $\Cap_1$, assuming elements weights are $\geq 1$), 
and  {\em Sum} -- the sum of weights of elements with keys in the
segment ($L_1$ moment or $\Cap_\infty$).

Frequency caps  that are in the mid-range mitigate the domination of  the
 statistics by the (typically few) very frequent keys but still provide
 a larger representation of the  more frequent  keys.
Mid-range frequency cap  statistics are prevalent in
online advertising platforms \cite{GoogleFreqCap,FacebookFreqCap}:   A common practice is to allow an advertiser to
 specify a limit to the number of impressions of an ad campaign that
 any individual user is exposed to in a particular duration of time.
Advertisements also typically target only a segment $H$ of users 
(say certain demographics and geographic location).  The statistics $Q(\Cap_T,H)$ is the number of
qualifying opportunities for placing an ad.  These queries are
posed over past data in order to provide an advertiser with a prediction for the
total potential number
of qualifying impressions.  Often, the prediction needs to be computed or
estimated quickly, to facilitate interactive campaign planning.

An exact computation of frequency statistics \eqref{basicquery:eq}
requires aggregating the data by key.
The representation size of the aggregated view, however, and the runtime state 
needed to produce it, are  linear in the number of distinct keys. 
Often, the number of distinct keys is very 
 large and our system can be using the same resources to process 
 many different streams or workloads.   
  To scalably mine such data,  our computation needs to be limited to
one or few passes over elements while maintaining a small runtime
 state (which translates to memory or communication). 
A single pass (stream computation) is necessary when the data is discarded  (such 
as with IP traffic) or when statistics are collected for live 
dashboards.  
Under these constraints, we often must settle for 
a small summary of the data set which
 can provide  approximate answers \cite{FlajoletMartin85,GM:sigmod98,ams99}.

   A solution which only addresses the final summary size 
is to first compute the aggregated view, and  
  then retain a sample for future queries:  For each key 
  we compute a weight equal to 
  $f(w_x)$, and we then apply a weighted sampling scheme
such as Probability Proportion to Size (pps) \cite{Tille:book},
VarOpt \cite{Cha82,varopt_full:CDKLT10}, or bottom-$k$,
which includes 
{\em  successive weighted sampling without replacement} (ppswor) and
sequential Poisson/Priority sampling \cite{Rosen1972:successive,Ohlsson_SPS:1998,DLT:jacm07}.

From the weighted sample  we can  compute {\em approximate}  segment frequency statistics  by applying an 
  appropriate {\em estimator} to the sample.  
 There is a well-understood tradeoff 
  between the sample size $k$ and the accuracy of our approximation. 
For a segment $H$ that has proportion  $q=Q(f,H)/Q(f,{\cal
  X})$ of the statistics value on the general key population, the
coefficient of variation (CV)  is (roughly) $(qk)^{-0.5}$: the inverse of the
square root of $qk$.   
The CV is the standard error normalized by the mean, and corresponds
to the NRMSE (normalized root mean square error).  That is, in order
to obtain NRMSE of $\epsilon=0.1$ (10\%)  on segments that have
at least $q=0.001$ fraction of the total value of the statistics, we need to
choose a sample size of $k=\epsilon^{-2}/q = 10^5$, which is usually
much smaller than the number of distinct keys we might have.  This
also means that we can obtain confidence intervals on our estimates using the actual number of
samples from our segment.
Moreover, this CV   bound is the best we can 
hope for (on average over segments) and will be the gold standard we use
in the design of sampling schemes and estimators in more constrained
settings, which preclude aggregating the data.


The challenge is to produce an effective sample of the data, using 
one or few passes while maintaining state that ideally is of the order
of the desired sample size.
  There is a large body of work on stream sampling schemes designed for
distinct and sum queries.
 The Sample and Hold (SH) family of sampling schemes
 \cite{GM:sigmod98,EV:ATAP02,CCD:sigmetrics12,flowsketch:JCSS2014}
 and another based on VarOpt \cite{incsum:VLDB2009} are suited
 for sum queries. Distinct reservoir sampling of keys
 \cite{Knuth2f,FlajoletMartin85}  is suited for distinct queries.
Both SH \cite{flowsketch:JCSS2014} and distinct sampling support unbiased estimates of all frequency
  statistics  and meet our $(qk)^{-0.5}$ CV upper bound target for
 the particular statistics they are designed for (the claim for SH is established here).
They do not provide, however, comparable statistical guarantees
  for other statistics.


\subsection*{Contributions and Road Map}

  Our main contribution is a general sampling framework for
  unaggregated streams.  The sampling scheme is specified by a
 random scoring function that is applied to stream elements.  The scoring
 function is tailored to the statistics we want to estimate.
Our framework is presented in Section \ref{framework:sec} and we
cast the existing distinct and SH sampling schemes as special cases.

Our framework facilitates the 
  first stream sampling solution with CV upper bound that is close to
  the $(qk)^{-0.5}$ gold standard for general  frequency cap statistics.    
  We offer two basic designs:  A discrete spectrum that only handles 
  uniform elements and a continuous spectrum which handles arbitrary 
  positive element weights.


 Our discrete spectrum is presented in Section \ref{discrete:sec}.
The sampling algorithms SH$_\ell$ are
 parametrized by an integer {\em cap} parameter $\ell$.  When $\ell$
 exceeds the maximum frequency over keys, SH$_\ell$ is equivalent to
classic SH.  For $\ell=1$, it is
 identical to distinct sampling.   
We derive unbiased and admissible estimators for 
any {\em discrete} frequency statistics, that is,  $f$  specified for nonnegative 
integers.

  Our continuous spectrum SH$_\ell$, for a positive real cap
  parameter $\ell$, is presented in Section~\ref{cont:sec}.
When $\ell \gg \max_x w_x$,   SH$_\ell$ is identical to
weighted SH \cite{CCD:sigmetrics12}.    For $\ell\ll \min_x w_x$,
SH$_\ell$ is distinct sampling.
We derive estimators of frequency
   statistics where the function $f$ is continuous and differentiable almost
   everywhere.  Note that most natural statistics, including frequency moments and 
   cap statistics, can be
   expressed as continuous monotone functions, which are
   differentiable almost everywhere.
Surprisingly perhaps, the continuous spectrum,
   which may seem less intuitive than the discrete spectrum, yields an
   elegant and simple  specification of estimators.

  We show that our  estimates of $\Cap_T$ statistics from SH$_\ell$
  samples have CV upper bounded by $O((qk)^{-0.5})$ when
 $T=O(\ell)$.  The CV bound gracefully degrades with
{\em disparity} $\max\{T/\ell,\ell/T\}$ between the sample cap parameter $\ell$ and the statistics cap parameter $T$.
The estimate of any frequency function $f$ is unbiased and for
$f$ that is monotone  non-decreasing, also nonnegative.
This makes our  design very versatile.

Our estimators are derived by expressing sampling as a
transform from frequencies to expected ``sampled frequencies,'' and then inverting the transform.
The transform is a matrix vector product In the discrete case and an
integral transform in the continuous case.
For the latter, the estimator is a 
simple expression in terms of $f$ and its derivative
$f'$.
Since our estimators are the unique inverse of the transform,
they are the minimum variance unbiased nonnegative estimators for the sampling
  scheme, meaning that in terms of variance,  they optimally use
the  information in the sample.  
Our discrete estimators generalize a matrix inversion applied in
\cite{hv:imc03,flowsketch:JCSS2014} to estimate the flow size distribution
from Sampled Netflow and SH IP flow records.  
Our continuous spectrum estimators are novel even 
for the basic weighted SH  scheme, for which previously only
estimators for sum statistics  were provided \cite{CCD:sigmetrics12}.

 In Section \ref{MO:sec} we \onlyinproc{propose a design for
   multi-objective sampling \cite{multiw:VLDB2009} that addresses}\notinproc{ address} applications that
 require estimates with statistical guarantees for multiple, possibly
 all, cap statistics. 
\notinproc{
One solution is to compute a set of samples with
different cap parameters which cover the range of statistics we are
interested in.  A  $\Cap_T$ statistics query can then be estimated
from the sample that has $\ell$ parameter closest to $T$.
We propose a design of a single  {\em
  multi-objective} sample that offers both more efficient sampling and
a better tradeoff of accuracy and sample size.
The design is based on our continuous spectrum and draws on
a multi-objective design for aggregated data \cite{multiw:VLDB2009} and 
the notion of sample coordination \cite{BrEaJo:1972}.
}

Our proposed sampling algorithms and estimators are simple and highly
 practical, despite a technical analysis.
The application resembles that of classic (uncapped) SH,  distinct sampling, and approximate distinct counting algorithms that
 are prevalent in industrial applications \cite{hyperloglogpractice:EDBT2013}.  Section \ref{experiments:sec}
 includes an experimental evaluation which demonstrates  superior
 accuracy versus sample size tradeoffs by using a sample that is
 suited for the statistic.
\onlyinproc{
 Due to a page limit, we could not include most proofs and many
 details and refer the reader to the technical report
 \cite{freqCap:2015}.
}

\section{Preliminaries} \label{prelim:sec}

\SetKwFunction{ElementScore}{ElementScore}
\SetKwFunction{Hash}{Hash}
\SetKwFunction{AdjustCountScore}{AdjustCountScore}
\SetKwFunction{rand}{rand}
We work with
key value data sets that consist of {\em elements} $(x,w)$, where $x$ is a key from a universe ${\cal
  X}$ and $w>0$.  The data set is {\em aggregated} if each key appears
in at most one element and is {\em unaggregated} otherwise.
We define the weight $w_x \equiv   \sum_{h \in
  x} w(h)$ of a key $x$ to be the sum of the weights
of elements  with key   $x$.  If $x$ is not active (there are no
elements with key $x$), we define $w_x=0$.
When element weights are uniform, we define $w_x$ to be the 
number of elements with key $x$.  
The {\em aggregated view} of an unaggregated data set has elements $(x,w_x)$ for all
active keys $x$.

  We are interested in sampling algorithms that process the
  unaggregated data in
  one or few passes while maintaining state that is proportional to
  the sample size.  
Such   algorithms can be scalably executed when elements are streamed
(presented sequentially  to the algorithm) or distributed across
multiple locations.

  We start with a quick review of relevant sampling schemes for aggregated data sets.
  A Poisson sample of a key value dataset  $\{(x,w_x)\}$ is
specified by sampling probabilities $p_x$.  The sample $S$ includes each
$x\in {\cal X}$ with independent probability $p_x$ and has
expected size $\E[|S|]=\sum_x p_x \equiv k$.  To estimate 
a frequency statistics $Q(f,H)$ from the sample, we can apply  the inverse 
probability estimator $\hat{Q}(f,H) = \sum_{x\in H\cap S}
f(w_x)/p_x$ \cite{HT52}.  
 This estimator can be interpreted as a sum of per-key estimates 
that are $f(w_x)/p_x$ if $x\in S$ and $0$ otherwise. 
Note that this estimator can only be applied when
$w_x$ and $p_x$ are available for all $x\in S$.  It is
nonnegative and is unbiased if $p_x   >0$ when $f(w_x)>0$.
It is actually the minimum variance 
 unbiased and nonnegative {\em sum} estimator (sum of per-key
 estimates) for the given probabilities $\{p_x\}$.

   For a dataset $\{(x,w_x)\}$, function $f$, and (expected) sample size $k$, 
one can ask what are the ``optimal'' sampling probabilities. 
  It is  well known that if we sample keys with probability 
  proportional to their contribution $f(w_x)$  (pps),
we minimize the sum of per-key variances $\sum_x f(w_x)^2 
  (1/p_x-1)$.    
With pps, we have  the following statistical guarantee:
For estimates of the statistics
$Q(f,H)$, where the segment $H$ has proportion
$$q =  \frac{Q(f,H)}{Q(f,{\cal X})} = \frac{\sum_{x\in H}
f(w_x)}{\sum_x f(w_x)} $$ of the statistics, the variance 
of our estimate is 
$$\var[\hat{Q}(f,H)] \leq \frac{1}{q k}Q(f,H)^2\ .$$   Thus the 
CV (normalized standard error) is at most $(q k)^{-0.5}$, which
is the best bound we can hope for on average
over segments with proportion $q$.  That is, any scheme that would do
better on some segments, would do worse on others.  Other weighted
sampling schemes we mentioned in the introduction provide this
statistical guarantee with a fixed
sample size $k$\onlyinproc{.}\notinproc{: VarOpt provides the
$(qk)^{-0.5}$ quality with better estimation for $q$ closer to $1$.
Sequential Poisson (Priority) sampling has $(q(k-1))^{-0.5}$ quality \cite{Szegedy05}.}
\ignore{
\footnote{An alternative scheme that obtains a fixed 
  sample size are {\em VarOpt sampling}, which realizes PPS sampling 
  probabilities and also has ``good'' correlations that improve 
  estimates for large segments \cite{Cha82,varopt_full:CDKLT10}.
  Another
class of sampling schemes that obtains a fixed sample size are
bottom-$k$ sampling, which draws for each key a random $\seed{x}$
which depends on $w_x$.  The sample then consists of the $k$ keys with
minimum seed and a threshold value $\tau$, which is the $(k+1)$ smallest
seed.   These schemes include Sequential Poisson (priority) sampling 
 \cite{Ohlsson_SPS:1998,DLT:jacm07, Szegedy:stoc06}, also has CV that
 is at most $(q(k-1))^{-0.5}$ \cite{Szegedy:stoc06} and also ppswor.}}

One of these schemes that is particularly relevant for our treatment of
unaggregated data sets is ppswor:
Keys are drawn successively so that at each 
step the probability that we draw $x$ is proportional to its weight 
relative to the remaining unsampled keys: $f(w_x)/\sum_{y\not\in S} f(w_y)$.
The sampling can be realized by associating with each key 
a random $\seed{x}\sim \Exp[f(w_x)]$ (exponentially
distributed seed with parameter $f(w_x)$)
\cite{Rosen1972:successive}.  
Ordering keys by increasing seed
value turns out to correspond exactly to ppswor sampling order.  
A  {\em fixed-threshold}
sample,  for a pre-specified threshold $\tau$, 
includes all keys with $\seed{x} < \tau$.
Alternatively, we can obtain a {\em fixed size} (bottom-$k$) sample by
taking the $k$ keys with smallest seed values.
In the latter case, it is convenient to define $\tau$ as the $(k+1)$ smallest
seed.

  Finally, we can estimate a statistics $Q(g,H)$ from the
ppswor sample taken for weights $f(w_x)$ as follows.  When we use fixed threshold
sampling, we compute the probability that $x$ is sampled
 $$\Phi_\tau(w_x)  \equiv 
\Pr[\seed{x}< \tau] = 1-e^{-f(w_x)\tau}\ ,$$ 
and apply inverse probability:
\begin{equation} \label{ppsworest}
\hat{Q}(g,H) = \sum_{x\in H\cap S} \hat{g}(w_x \mid \tau),\,
\text{where}\, 
\hat{g}(w_x \mid \tau) \equiv \frac{g(w_x)}{\Phi_\tau(w_x)}\ . 
\end{equation}
Note that $\Phi_\tau(w_x)$ only depends on $w_x$ and $\tau$  (which are 
available for sampled keys). 
When we work with a fixed sample size $k$ and define $\tau$ to be the
$(k+1)$ smallest seed, 
we can interpret $\Phi_\tau(w_x)$ as the 
probability that 
the  key $x$ is sampled, conditioned on fixed randomization 
of other keys.  This means that the estimator \eqref{ppsworest} is unbiased
 \cite{bottomk:VLDB2008}.  
Moreover, the estimates $\hat{g}(w_x \mid \tau)$ obtained for
different keys $x$ have
zero covariances \cite{bottomk:VLDB2008}, which allows us to bound the
variance on segment queries as we would do when sampling with a pre-specified threshold.
It turns out \notinproc{(see Theorem \ref{ppsworcv:thm})}\onlyinproc{(see TR\cite{freqCap:2015})}
that for statistics $\hat{Q}(f,H)$ with proportion $q$, the
CV is at most $(q(k-1))^{-0.5}$, which is essentially (within a single
sample) our ``gold standard'' CV.


A ppswor sample with respect to any function $f(w_x)$
can be computed from a streamed (or distributed)
aggregated data $\{(x,w_x)\}$, using state proportional to the sample size.
This is not generally possible, however, over unaggregated data: For example, there are polynomial
lower bounds on the state needed by a streaming algorithm which approximates frequency moments
$Q(x^p,{\cal X})$ with $p>2$ \cite{ams99}.

\section{Sampling framework} \label{framework:sec}

  We present a framework for sampling unaggregated data sets and
  cast SH and distinct sampling in our framework.
Our algorithms compute a sample $S$ while 
maintaining state, in the form of a cache $S$ of sampled keys, that is 
proportional to the sample size.  
  Each sampling scheme in specified through a random mapping
$\ElementScore{h}$ of elements $h=(x,w)$ to numeric score values.
The distribution of \ElementScore may only depend on the key $x$ and $w$.
We then define the {\em seed} of a key $x$
\begin{equation}\label{seed:Eq}
\seed(x) = \min_{h \text{ with key }  x} \ElementScore{h}
\end{equation} to be  the random variable that is the minimum score of all its elements.  
 As with ppswor, we can obtain a fixed-threshold sample $S = \{x \mid
 seed(x)<\tau \}$, which for a
 given $\tau$  includes all keys   with $\seed(x) < \tau$,  or a
 fixed-size sample, which for a specified  sample size $k$
includes the $k$ keys with smallest seed values and define
$\tau$ to be the $(k+1)$ smallest seed.

   Once we have the sample, we can apply estimators to it to approximate
statistics.  To do so, we need to
   have information on the weight  of sampled keys.
The exact weights $w_x$ of sampled keys $x\in S$  can be computed in a second pass over the data, as we detail below.
We also consider a pure streaming (single sequential pass) setting, where we 
generally settle for
some $c_x \leq w_x$ and we derive estimators that are able to work with 
this information.  
In terms of computation platform, our 2-pass schemes can be fully
parallelized or distributed whereas
our 1-pass (streaming) schemes are not as flexible: They
can be executed on multiple streams that are processed
separately (as with sharding) provided that  all
elements with the same key are processed  at  the same shard.

\subsection{2-pass scheme} \label{2pass:sec}
The first pass identifies 
the set of keys $S$ with smallest seeds.  
For fixed-threshold sampling, our summary contains all keys 
  with scores below $\tau$.  With  fixed-size sampling, the summary
  contain the keys with $k$ smallest minimum scores.
 These summaries are mergeable, that is, from the summaries of
two data sets we can compute a summary of their union.  For
fixed-$\tau$, the merged summary is the union of the keys in the two summaries.  For
fixed-$k$, we compute the seed of each key in the union as the
minimum seed attained in each of the summaries.  We then take the $k$
keys with smallest seeds (retaining their seed values) as a summary of
the union.  Either way the summary sizes never exceed $|S|$, which is
the final sample size.  The second pass, which computes $w_x$ for
$x\in S$ uses summaries that are the weight of each key $x\in S$ the
data set.  We merge two summaries by key-wise addition of weights to
obtain a summary of the union.
  Algorithm \ref{2pass:alg} is 2-pass stream sampling of a fixed sample size
  $k$.   \notinproc{Simple variations handle distributed or parallel
  computation or fixed threshold sampling.}

\begin{algorithm2e}[h] \label{twopassdiscrete:alg}
\caption{2-pass stream sampling: fixed size $k$  \label{2pass:alg}}
\DontPrintSemicolon 
\SetKwArray{Counters}{Counters}
\SetKwFunction{rand}{rand}
\SetKwFunction{Hash}{Hash}
\KwData{sample size $k$, elements $(x,w)$ where $x\in
  {\cal X}$ and $w>0$}
\KwOut{set of $k$ pairs $(x,w_x)$ where $x\in {\cal X}$}
\SetKwFunction{Return}{return}
$\Counters \gets \emptyset$ \tcp*[h]{Initialize sample} \;
$\tau \gets +\infty$ \tcp*[h]{Upper bound on \ElementScore} \;
\tcp{Pass I: Identify the $k$  sampled keys}
\ForEach{stream element $h=(x,w)$}
{
  \If{$x$ is in $\Counters$}
  {
    $\seed{x} \gets \min\{\seed{x}, \ElementScore{h}\}$\;
  }
  \Else 
  {
    $s \gets \ElementScore{h}$\;
    \If{$s < \tau$}
    {
      $\seed{x}\gets s$; $\Counters{x} \gets 0$\;
      \If {$|\Counters| = k+1$}
      {
        $y \gets \arg\max\{\seed{x} \mid x \text{ in } \Counters\}$\;
        $\tau \gets \seed{y}$\;
        delete $\seed(y)$, $\Counters{y}$\;
      }
    }
  }
}
\tcp{Pass II:  Compute $w_x$ for sampled keys}
\ForEach{stream element $h=(x,w)$}
{
  \If{$x$ is in $\Counters$}
  {
    $\Counters{x} \gets \Counters{x}+w$
  }
}
\Return{$\tau$ ;  $(x,\Counters{x})$ for $x$ in $\Counters$}
\end{algorithm2e}

\subsection{Fixed threshold stream sampling}
 A fixed threshold scheme processes an element $h=(x,w)$ as follows.
If $x\in S$ (key $x$ is cached/sampled), then $c_x \gets c_x+w$.  Otherwise,
if $\ElementScore{h} < \tau$, then $x$ is inserted to $S$ and
$c_x \leq w$ is initialized.  The discrete scheme which applies to
uniform weights $w=1$, is provided as  Algorithm \ref{ellSHth:alg}, and
uses the initialization  $c_x\gets 1$ ($\Counters{x}$ in the
pseudocode).
A continuous scheme is presented in Section \ref{cont:sec}.

\begin{algorithm2e}[h]
\caption{Stream sampling with fixed threshold $\tau$ \label{ellSHth:alg}}
\DontPrintSemicolon 
\SetKwArray{Counters}{Counters}
\SetKwFunction{rand}{rand}
\SetKwFunction{Hash}{Hash}
\SetKwFunction{Return}{return}
\SetKwFunction{ElementScore}{ElementScore}
\KwData{threshold $\tau$, stream of elements with key $x\in {\cal X}$}
\KwOut{set of pairs $(x,c_x)$ where $x\in {\cal X}$ and $c_x \in [1,w_x]$}
$\Counters \gets \emptyset$  \tcp*[h]{Initialize $\Counters$ cache}\; 
\ForEach(\tcp*[h]{Process a stream element}){stream element $h$ with
  key $x$}
{
\If{$x$ is in $\Counters$}{$\Counters{x} \gets \Counters{x} +1$;}
\Else{
 \If{$\ElementScore{h}  <
   \tau$}{$\Counters{x} \gets 1$;
 \tcp*[h]{Initialize $c_x$ }}
}
}
\Return{$(x,\Counters{x})$ for $x$ in $\Counters$}
\end{algorithm2e}

\subsection{Fixed size stream sampling}

Algorithm~\ref{ellaSH:alg} provides pseudocode for discrete (uniform
weights) stream sampling with a fixed sample size $k$.

  The algorithm maintain a set $S$ ($\Counters$) of cached keys.  For
  each cached key $x$, it keeps a count $c_x$ ($\Counters{x}$) and a
  lazily computed seed value $\seed(x)$.
When processing an element $h$ with key $x$, we compute
$y\gets \ElementScore{h}$.  If $x\in S$, we increment $c_x$.

Otherwise, if $x\not\in S$ and $y< \tau$, we 
insert $x\in S$ with $c_x\gets 1$ and $\seed(x) \gets y$.  
As a result, we may have
$|S|=k+1$ cached keys.  In this case, we would like to evict from $S$ the key 
with maximum seed.  But the seeds are not fully evaluated yet, in that
the current $\seed(x)$ only reflect the seed up to the first element
that is currently counted in $c_x$.

We repeat the following until a key is evicted.
We  pop from $S$ the key $y$ with maximum 
current seed and set  $\tau \gets \seed(y)$.  
We then iterate decreasing the count $c_y$ and scoring ``uncounted'' elements until
either the count becomes $c_y=0$ and $y$ is evicted or we obtain a
score that is below $\tau$.
 In the latter case, we reinsert $y$ to $S$ with 
$\seed(y)$ equal to that score.

\begin{algorithm2e}[h]
\caption{Stream sampling with fixed size $k$  \label{ellaSH:alg}}
\DontPrintSemicolon 
\SetKwArray{Counters}{Counters}
\SetKwFunction{rand}{rand}
\SetKwFunction{Hash}{Hash}
\SetKwFunction{AdjustCountScore}{AdjustCountScore}
\KwData{sample size $k$, stream of elements with key $x\in {\cal X}$}
\KwOut{set of $k$ pairs $(x,c_x)$ where $x\in {\cal X}$ and $c_x \in [1,w_x]$}
\SetKwFunction{Return}{return}
$\Counters \gets \emptyset$  \tcp*[h]{Initialization} \;
$\tau \gets 1$ \tcp*[h]{Supremum of \ElementScore range} \;
\ForEach{element $h$ with key $x$}
{
  \If{$x$ is in $\Counters$}{$\Counters{x} \gets \Counters{x}+1$\;}
  \Else{
    $score \gets \ElementScore{h}$ \;
    \If{$score < \tau$}
    {
      $\seed{x} \gets score$\; 
      $\Counters{x} \gets 1$\;
      \While {$|\Counters| > k$}
      {
        $y \gets \arg\max\{\seed{x} \mid x \text{ in } \Counters\}$\;
        $\tau \gets \seed{y}$\;
        \While{$\Counters{y} > 0$ and $\seed{y} \geq \tau$}
        {$\Counters{y} \gets \Counters{y}-1$\;
          $\seed{y} \gets \ElementScore{y}$\;      } 
        \If{$\Counters{y} == 0$}{ delete $\Counters{y}$, $\seed{y}$}\;
      }
    }
  }
}
\Return{$\tau$ ;  $(x,\Counters{x})$ for $x$ in $\Counters$}
\end{algorithm2e}

\paragraph*{Analysis}
Clearly, the work of Algorithm \ref{ellSHth:alg} (fixed-threshold sampling) is $O(1)$ per
stream element.  \onlyinproc{We show amortized work $O(1)$ 
  for  Algorithm \ref{ellaSH:alg}  (fixed-size sampling).  (See TR \cite{freqCap:2015} for proof.)}
\notinproc{
We show the following for fixed-size sampling:
\begin{lemma}
The amortized per-element work of  Algorithm \ref{ellaSH:alg}  is $O(1)$.
\end{lemma}
\begin{proof}
 The algorithm maintains at most $k$ cached keys in a max priority queue,
  accessible by decreasing  $\seed(x)$. When there are fewer than $k$
  active keys, all of them are cached, and otherwise $k$ keys are 
  cached.   The costlier operations are 
eviction steps, which happen when a new key is inserted and the cache 
  is full.  The expected total number of evictions is at most $k\ln
  m'$,  where $m'$ is the expected number of distinct element scores. 
The value of $m'$ depends on our element scoring function but is
always at most the number of elements and at least the number
of distinct keys (since scores of different keys are independent).
In any case, the number of evictions is logarithmic in the stream size.


  In an eviction step, an element $x$ is popped 
and $c_x$ is decremented at least once.  If the final count is not zero, then 
  $x$ is placed back on the queue with a strictly lower $\seed{x}$. 
  The median value of the new $\seed{x}$ in this case is the
  median of the score distribution, provided it is lower than $\tau$.

  The total work on decreasing the 
  counts can be ``charged'' to the processing of the corresponding element, but there is also a 
  possible charge of a priority queue insertion for keys whose count got 
  decreased and did not get removed.  A priority queue operation cost
  is about $O(\log k)$.  We can bound the number of such operations by
  noting that in expectation, $\seed(x)$  decreases so that in
  expectation the  probability of a new element score being below it
  is  halved.   Which means that the expected number of times a key
  can be placed back is at most logarithmic in the number of distinct
  scores its elements can have.  It also means that $k$ ``place
  backs''  corresponds in expectation to a decrease of the threshold to the
  conditional median, which can happen when $m'$ doubles. So in
  expectation there are $O(1)$ place backs per eviction step.
\end{proof}
}

\subsection{Element scoring properties} \label{escoreprop:sec}

We will select the element scoring function according to the
 statistics $f$ we are interested in.  Intuitively, to obtain quality estimates
 (CV upper bound of $(qk)^{-0.5}$), we would need  to
sample each key $x$  with probability roughly proportional to
 $f(w_x)$.  The challenge is to identify when and how we can
 achieve this by a small state streaming algorithm.

 Some properties of our element scoring functions that greatly
 simplify the derivation of estimators are that
$\seed$ values of different keys are independent and that for a particular key $x$,
the distribution of $\seed(x)$ (the minimum element score) depends only
on $w_x$, and not on the arrangement of elements or on the breakdown
  of the weight of each key to different elements.  
Furthermore, we would also want 
the distribution of $c_x$ for $x\in S$ to only depend on $w_x$ and
$\tau$.
We assume here that we work with perfectly random 
numbers and hash functions.

\subsection{Estimation}
As with the ppswor 
estimator reviewed in Section \ref{prelim:sec}, we use estimators that
can be expressed as
a sum over keys $x\in H$ of individual estimates $\hat{f}(w_x)$ of
$f(w_x)$.  The estimate are unbiased and are $0$ for keys $x\not\in S$.

With two-pass sampling (Algorithm \ref{2pass:alg}), we have the weight
$w_x$, and therefore $f(w_x)$, for each sampled key $x\in S$.  
When the seed distribution only depends on $w_x$ and $\tau$, we can compute
the inclusion probability $\Phi_\tau(w_x)$ of a key $x$ from its weight $w_x$
and apply inverse probability estimation as in \eqref{ppsworest}.
In the streaming (single pass) schemes, the sample includes a
partial count $c_x\leq w_x$ for each $x\in S$.  The requirement that
the distribution of $c_x$ only depends on $w_x$ and $\tau$ allows us to 
express sampling 
as a  transform (which depends on $\tau$) from the distribution $w_x$ to the
expected outcome distribution $c_x$.  The derivation of unbiased 
estimators then corresponds to inverting this transform.

  The transforms we obtain have a unique inverse, which means that our
estimators are
the optimal (minimum variance) unbiased and nonnegative sum
estimators.
Because the 2-pass estimators \eqref{ppsworest} are also optimal, and rely on more
information -- the exact value $w_x$ instead 
of a sample from a distribution with parameter $w_x$, the variance 
of the streaming estimators is always at least that of the  2-pass estimator.

The estimators for both the fixed-threshold and the
 fixed sample-size schemes are stated in terms of the threshold
 probability $\tau$. When working with a fixed sample-size $k$, $\tau$ is
 defined as the $(k+1)$st smallest seed.  As with ppswor, 
$\tau$  when defined this way 
plays the same role as the threshold value $\tau$ used in a fixed sampling threshold
  scheme: The probability that a key is sampled, conditioned on fixed randomization
of other keys, is the probability  that its seed value is below the
$k$th smallest seed of other keys.  When the key is
included in the sample, this value is $\tau$.  Similarly, under the same conditioning, 
the distribution of $c_x$ only depends on $\tau$ and
$w_x$,  and is the same one as the respective fixed threshold scheme
with $\tau$.  
Moreover,  the covariance of the
estimates obtained for two different keys $x,y$ is zero.  The argument
is the same as with ppswor \cite{bottomk:VLDB2008} and SH \cite{flowsketch:JCSS2014}. This important
property allows us to bound the variance on estimates of segment
statistics by the sum of variance of estimates for individual keys.

  We now cast two existing basic sampling schemes in our framework:
Distinct, which 
is designed for $\Cap_1$ statistics and
SH, which is designed for $\Cap_\infty$ (sum) statistics.

\subsection{Distinct sampling}  A distinct sample is a uniform
sample of active keys (those with $w_x>0$), meaning that conditioned
on sample size $k$, all subsets of active keys are equally likely.
For an element $h$ with key $x$, we use 
$\ElementScore{h} = \Hash{x}$, where $\Hash{x} \sim U[0,1]$ is a 
random hash function selected before we process the stream.
Note that all elements of the same key $x$ have the same score and 
therefore $\seed(x) \equiv \Hash{x}$.

When we sample with respect to a fixed threshold $\tau$, we retain
  all keys with $\Hash{x} < \tau$. When using a fixed
  sample size $k$, the scheme is
the following (distinct variant) of reservoir sampling
\cite{Knuth2f}:   For each stream element, compute $\Hash(x)$ and
retain the $k$ keys with smallest hash values.  


 With distinct sampling, the value $c_x$ is equal to the exact weight
$w_x$ for each sampled key $x$.  This is because any key that enters
our cache does so on the first element of the key.  If a key is evicted,
(in the fixed $k$ scheme), it can never re-enter.   We also have that
for all keys with $w_x>0$, the probability that $x$ is sampled is
$\Phi_\tau(w_x) \equiv \tau^{-1}$.
We can therefore apply the  inverse probability estimator
\eqref{ppsworest}:
\begin{equation} \label{distinctest}
\hat{Q}(f,H) = \tau^{-1} \sum_{x\in S\cap    H} f(w_x) \ .
\end{equation}

Distinct sampling is optimized for distinct ($\Cap_1$) statistics.  In particular,
$\hat{Q}(\Cap_1, {\cal X})$ has CV upper bounded by $(k-1)^{-0.5}$
\cite{ECohen6f,ECohenADS:PODS2014} and for  a segment $H$ with proportion
$q$ of distinct keys, $\hat{Q}(\Cap_1,H)$ has CV upper bounded by $(q(k-1))^{-0.5}$, as it is
  equivalent to the ppswor estimator for $f(w)\equiv 1$.
For general $\Cap_T$ statistics, however, the CV grows rapidly with
$T$ (we shall see it is $\propto \sqrt{T}$).
 This is because our
uniform sample of active keys can easily miss 
keys with high $f(w_x)$ values which contribute more to the statistics.

\onlyinproc{\newpage}
\subsection{Sample and Hold (SH)}  Classic SH, with fixed sampling
threshold $\tau$ or with fixed sample size $k$,
\cite{GM:sigmod98,EV:ATAP02} is specified for uniform 
element weights, so that $w_x$ is the number of elements with key $x$.
We cast SH in our framework 
using 
$\ElementScore{h} \sim U[0,1]$.  Note that  each key $x$ can have  many independent scores drawn, one
 for each element of $x$. Therefore, the more elements a key has, the more 
  likely it is to be sampled.  The seed is the minimum element score,
  which can be transformed to an exponentially distributed random
  variable with parameter $w_x$.  Therefore, as observed in
  \cite{flowsketch:JCSS2014},  the SH sample is actually a
  ppswor sample with respect to the weights $w_x$
  \cite{Rosen1972:successive}.  
When we use  a second pass (Section \ref{2pass:sec}) to obtain the
exact weights $w_x$, we can apply the ppswor estimator \eqref{ppsworest}.

 With stream sampling (Algorithm~\ref{ellSHth:alg} and Algorithm~\ref{ellaSH:alg}),
the final count of a key $x$ has $c_x \leq w_x$, where 
$w_x- c_x+1$ is geometric  with parameter
  $\tau$, truncated at $w_x+1$ (probability of $c_x=0$ is $(1-\tau)^{w_x}$).
An unbiased estimator for  statistics $Q(f,H)$ from an SH sample is \cite{flowsketch:JCSS2014}:
\notinproc{\footnote{Estimators for a related scheme  (where 
  elements are drawn with replacement) were presented in  \cite{ams99}.}}
\begin{equation} \label{pureSHest}
\hat{Q}(f,H) = \tau^{-1} \sum_{x \in S\cap H} \bigg(f(c_x) - f(c_x-1)(1-\tau)\bigg)\ . 
\end{equation}
\notinproc{\footnote{With fixed-size sampling, we can instead use here the 
stratified   value $\tau =k/(k+\sum_{x\in {\cal X}} w_x - \sum_{x\in S\cap {\cal X}} c_x)$.}}
Note that this estimator is nonnegative when $f$ is monotone non 
decreasing.  This is  because for all $i > 0$, $f(i)-f(i-1)(1-\tau) >
0$.  Surprisingly perhaps, we show here \onlyinproc{(see TR \cite{freqCap:2015})}\notinproc{(Theorem \ref{aSHcv:thm})}
that the 1-pass estimate is not too far from the 2-pass estimate in that
for sum statistics  ($f(x)=x$) the CV is also upper bounded by $(q(k-1))^{-0.5}$.

 For cap statistics with small $T$, however, the SH estimates can have
 CV that far exceeds our $(qk)^{-0.5}$  target:  When the frequency distribution is
 highly skewed,  the ppswor sample would be dominated by heavy
  keys. This means that segments with a large proportion of the $\Cap_T$
  statistics that mostly include keys with low frequencies
would have a disproportionally small representation in the sample and
thus large errors.

\begin{figure}[htbp]
\centerline{
\resizebox{0.4\textwidth}{!}{\input{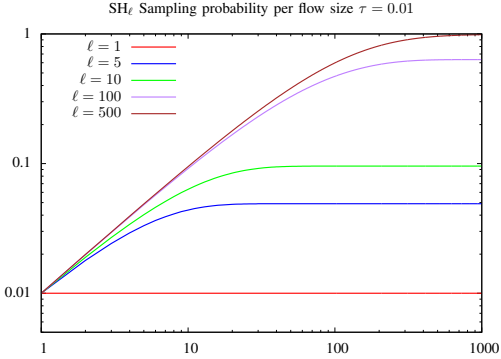}}
}
\caption{SH$_\ell$ sampling probability per key weight $w$, for
 selected values of $\ell$  ($\tau=0.01$).
 Note that for $w \gg \ell\log \ell$ probability is constant and 
for $w \ll \ell$,  probability is
proportional to $w$.  We can see that the probability is close to
being proportional to $\min\{w,\ell\}$, which is what we want for estimating
$\Cap_\ell$ statistics.
 \label{phicum:fig}}
\end{figure}

\section{The discrete SH spectrum} \label{discrete:sec}

 Our discrete SH spectrum is parametrized by an integer $\ell \geq 1$.  
Distinct sampling is SH$_1$ and classic SH is SH$_\infty$.   In general, SH$_{\ell}$ is designed
  to estimate well frequency cap statistics with $T \approx \ell$.

  The SH$_\ell$ element scoring function for an element $h$ with key $x$ 
draws a uniform random bucket 
$b \sim U[1,\cdots,\ell]$ and returns a hash of the pair 
$\Hash(x,b) \sim U[0,1]$.
Note that  the buckets are independent for different elements with
key $x$.
\begin{equation}\label{delementscore:eq}
\ElementScore{h} \gets \Hash(\lfloor(\ell*\rand{})\rfloor,x)\ .
\end{equation}
Recall that $\seed(x)$  \eqref{seed:Eq} is  the minimum
score of an element with key $x$.
When $\ell=1$, the
seed distribution is uniform for all keys with $w_x>0$.
More generally, we can see that the element scoring \eqref{delementscore:eq}
provides  up to $\ell$ ``independent'' attempts for each key to obtain
a lower seed. That way,  keys with more elements are more likely to have a lower seed
and be sampled, but with diminishing return:  Keys where  $w_x \ll
\min\{\ell,\tau^{-1}\}$ are sampled with probability roughly
proportional to $w_x$ whereas keys with $w_x \gg \min\{\ell,\tau^{-1}\}$ have a roughly constant 
inclusion probability regardless of frequency.  Also note that when
the cap parameter is large relative to the inverse sampling threshold $\ell \gg \tau^{-1}$, SH$_\ell$  is similar to
SH$_\infty$.
 Figure~\ref{phicum:fig} illustrates these properties by showing the sampling probability of a key 
as a function of $w_x$, for selected values of the  parameter $\ell$.

\ignore{
\begin{gnuplot}[terminal=tikz,terminaloptions=font ”10” linewidth 3]
set logscale x
set logscale y
set style data lines
set key left top
set title 'SH$_\ell$ Sampling probability per flow size $\tau=0.01$'
plot [:] [0.005:] 'phiplot_0.01_1000.data' using 1:3 title '$\ell = 1$' lw 3 lc 'red', 'phiplot_0.01_1000.data' using 1:5 title '$\ell = 5$' lw 3 lc 'blue', 'phiplot_0.01_1000.data' using 1:7 title '$\ell = 10$' lw 3 lc 'green', 'phiplot_0.01_1000.data' using 1:9 title '$\ell = 100$' lw 3 lc 'purple', 'phiplot_0.01_1000.data' using 1:11 title '$\ell = 500$' lw 3 lc 'brown'
\end{gnuplot}
}

\subsection{Estimators for discrete SH$_\ell$} \label{discreteSHest:sec}

  The output of our stream sampling algorithm is a threshold 
value $\tau$ and a set $S$ of pairs of the form 
  $(y,c_y)$, where $y\in {\cal X}$ and $c_y \in [1,w_y]$.

\smallskip
\noindent
{\bf Coefficient form.}
  We express our estimators as vectors 
 $\vecbeta^{(f,\tau,\ell)}$, which depends on $f$, the threshold
 $\tau$, and the parameter $\ell$.  The $c$th entry
 $\vecbeta^{(f,\tau,\ell)}_c$
is the contribution to the estimate of a key with count $c$. The
estimate on the statistics $Q(f,H)$ is then
\begin{equation} \label{coefform:eq}
 \hat{Q}(f,H) = \sum_{x \in S\cap H} \beta_{c_x} \ .
\end{equation}

  The distinct sample ($\ell=1$) estimator \eqref{distinctest} is
  expressed using $\beta_i  \equiv  f_i \tau^{-1}$  (using the notation
  $f_i\equiv f(i)$) whereas the SH
  estimator ($\ell = +\infty$)  \eqref{pureSHest}  is expressed using $\beta_i \equiv \tau^{-1}\bigg( f_i-f_{i-1}(1-\tau)
\bigg)$.   We seek estimators of this form for general
$\ell$ that are unbiased,  admissible, and nonnegative
$\vecbeta \geq 0$ when $f$ is non-decreasing.

\smallskip
\noindent
{\bf Probability vector $\vecphi$.}
Let $\phi_i$ be the
  probability that the $i$th element of the same key was the first one to get
  counted by SH$_\ell$.  The vector $\vecphi$ depends on the parameters $\ell$ and $\tau$.

  For $\ell=1$, we have the closed form $\phi_1 \equiv \tau$ and $\phi_i = 0$ for
  $i>1$.  For $\ell = +\infty$, we have $\phi_i = (1-\tau)^{i-1}\tau$.

  To express $\vecphi$ for general $\ell$,  we
 let $a_{ij}$ be the probability that we used exactly
  $j\leq \min\{\ell,i\}$ buckets in the first $i$ elements of a key.

By definition $a_{0i}\equiv 0$ when $i\geq 1$,
$a_{ij}\equiv 0$ when $j > \min\{\ell,i\}$, and $a_{1,0}= 0$.
Otherwise,  $a_{1,1}=1$ and for $i>1$, $j\leq
\min\{\ell,i\}$, the values can be computed from the relation 
\begin{equation} \label{avalues}
a_{ij} = a_{i-1,j}\frac{j}{\ell}+ a_{i-1,j-1}\frac{\ell-j+1}{\ell}\
. 
\end{equation}
\notinproc{
Note that as $i$ grows, the vectors $a_{i\cdot}$ converge to a vector that has all
entries $0$ except  $a_{i\ell}=1$.  It therefore suffices to
compute these entries only until $i=O(\ell \log(\ell))$.
For larger values of $i$ we can use the vector $a_{i\cdot}=(0,\ldots,0,1)$.

} We can now write
$$ \phi_i = \tau \sum_{j=1}^{\min\{i-1,\ell-1\}} a_{i-1,j}(1-\tau)^j
\frac{\ell-j}{\ell}\ .$$ \notinproc{
  Note that it always suffices to compute only the $M$ first entries 
  of $\phi$, where $$M = O(\min\{ \ell \log \ell, \tau^{-1} \log 
  \tau^{-1} \})\ .$$}

\smallskip
\noindent
{\bf A 2-pass estimator.} 
  The probability that a key $x$ is sampled (illustrated in Figure
  \ref{phicum:fig})  is
$$\Phi_{\tau,\ell}(w_x) \equiv \sum_{j=1}^{w_x} \phi_j\ .$$ 
If we use a 2-pass scheme (Section \ref{2pass:sec}), we can apply
the inverse probability estimator \eqref{ppsworest} $\hat{Q}(f,H) = \sum_{x\in S\cap H} \frac{f(w_x)}{\Phi(w_x)}$.

\ignore{
 The inverse probability estimator, which extends the ppswor estimator 
 \eqref{ppsworest} is 
\begin{equation} \label{invprobest:eq}
\hat{Q}(f,H) = \sum_{x\in S\cap H} \frac{f(w_x)}{\Phi(w_x)}\ . 
\end{equation}
}

\smallskip
\noindent
{\bf Inverting the sample counts.}
We now derive a streaming estimator.
We use the notation
$o_i = \{x\in S\cap H \mid c_x=i\}$  (the ``observed'' count) 
for the random variable that  is
the number of keys  $x\in S\cap 
H$ with $c_x = i$.  Let $m_i = \{x\in H \mid w_x=i \}$  be the
number of keys in $H$ with count $w_x=i$. 
Our statistics \eqref{basicquery:eq} can be expressed as $Q(f,H)=\vecf^{T} \vecm$.  
 We have the relation
$\E[o_i] = \sum_{j\geq i} \phi_{j-i+1} m_i$ 
and can write
$$\E[\veco] = Y^{(\vecphi)} \vecm\ .$$
We use the notation $Y^{(\vecv)} $ for an
upper triangular  matrix that corresponds to a vector $\vecv$, such
that $\forall ,j \geq i$, $[Y^{(\vecv)}]_{ij}   \equiv v_{j-i+1}$.

We have $\vecm = (Y^{(\vecphi)})^{-1} \E[\veco]$.   Therefore, from linearity, 
$\hat{\vecm} \equiv (Y^{(\vecphi)})^{-1} \veco$ is an
 unbiased estimator of $\vecm$.  Therefore, to compute the estimate we
 need to invert $Y^{(\vecphi)}$.

  The inverse of the matrix $Y^{(\vecphi)}$ has the same upper triangular
  structure, and can be expressed as   $Y^{(\vecpsi)}$  with   respect to another vector $\vecpsi$.
To compute $\vecpsi$, we consider the constraints  $Y^{(\vecpsi)}
Y^{(\vecphi)} = I$ obtained from the product of the first row of
$Y^{(\vecpsi)}$ with the columns of $Y^{(\vecphi)}$.
We obtain the equations
$\psi_1 = \phi_1^{-1}$,  and
  for $j>1$,
$$\sum_{j=1}^i \psi_j \phi_{1+i-j} = 0\ .$$
This allows us to iteratively solve for $\psi_i$ after computing $\psi_j$
for $j<i$ using
 $$\psi_i = \phi_1^{-1} (-\sum_{j=1}^{i-1} \phi_{1+i-j} \psi_j)\ .$$
For distinct sampling we have  $\psi_1 = \tau^{-1}$ and 
 $\psi_i=0$ for $i>1$.  For SH 
 \cite{flowsketch:JCSS2014} we have
$\psi_1 = \tau^{-1}$, $\psi_2 =
 -(1-\tau)\tau^{-1}$, and $\psi_i = 0$ for $i\geq 2$. In general,
 however, $\vecpsi$ can  have many non-zero entries.

 We show the following \onlyinproc{(see TR \cite{freqCap:2015})}:
\begin{theorem} \label{discreteestcoef:thm}
The estimator $\hat{Q}(f,H)=\sum_{x\in S\cap H} \beta_{c_x}$ , where
$$\beta^{(f,\tau,\ell)}_i \equiv \sum_{j=1}^i \psi_j f_{i-j+1}\ $$
is unbiased.
\end{theorem}
\notinproc{
\begin{proof}
 By substituting $\hat{\vecm} = Y^{(\vecpsi)} \veco$ in 
 $Q(f) = \vecf^T\vecm$, we obtain the estimator 
$\hat{Q}(f) = \vecf^T Y^{(\vecpsi)} \veco$.

  The unbiased estimate for $m_i$ is
$$\hat{m}_i = \sum_{j \geq i}  o_j  \psi_{j-i+1}\ .$$
The unbiased estimate for the contribution of keys with $i$ elements to the
statistics is
$$f_i \hat{m}_i = f_i \sum_{j \geq i}  o_j  \psi_{j-i+1}\ .$$

Therefore, the total contribution, and expressed in terms of $o_i$ is
$$\sum_i \sum_{j\geq i}  f_i o_j  \psi_{j-i+1} = \sum_i o_i \sum_{j=1}^i \psi_j f_{i-j+1}\ .$$
\end{proof}
Since the inverse is unique, our estimator is the only unbiased
estimator of this form and thus also admissible (minimum variance of
this form).
}
\notinproc{
 Note that when we only compute the first $M$ entries of $\psi$, we
limit the sum expression to range from $1$ to $\min\{M,i\}$.   In
applications, the coefficients $\beta$ only need to be computed for
$i$ such that there is at least one key $x$ in the sketch with $c_x = i$.
}

  We show that the estimates are nonnegative when 
$f$ is monotone non-decreasing:
\begin{theorem}
When $f$ is monotone non-decreasing, then for all $\ell$ and $\tau$,
$\vecbeta^{(f,\tau,\ell)} \geq 0$.  
\end{theorem}
\notinproc{
\begin{proof}
 We first claim that any prefix sum of $\vecpsi$ is positive.  That is,
\begin{equation} \label{psiprefix:eq} 
\forall j\geq 1,\ \sum_{i=1}^j \psi_i > 0\ .
\end{equation}
We prove the claim  by induction on $i$.
The base case of the induction has $\psi_1 = \phi_1 ^{-1} \equiv
\tau^{-1} > 0$.
We now show that $\sum_{i=1}^h \psi_i > 0$ if the claim \eqref{psiprefix:eq}
holds for all $j<h$.
We have 
\begin{eqnarray*}
0 &=& \sum_{j=1}^h \psi_j \phi_{h-j+1} \\
 & = & \phi_1 \sum_{i=1}^h \psi_j
+ \sum_{j=1}^{h-1} (\sum_{i=1}^j \psi_i) (\phi_{h-j+1} - \phi_{h-j})\
.
\end{eqnarray*}
Rearranging, we obtain
\begin{eqnarray*}
\phi_1 \sum_{i=1}^h \psi_j &=&
 \sum_{j=1}^{h-1} (\sum_{i=1}^j \psi_i) (\phi_{h-j}-\phi_{h-j+1})\ .
\end{eqnarray*}
We now argue that the right hand side is nonnegative.  In fact, each
summand, and each term in the product are nonnegative.  
Nonnegativity of the sums $\sum_{i=1}^j \psi_i$ follows from the
induction hypothesis for $j<h$.  
Nonnegativity of the differences $\phi_{h-j}-\phi_{h-j+1}$ for $j<h$ 
follows from
$\phi_i \geq 0$ being non-increasing (recall that 
$\phi_i$ is the probability that the $i$th element of a 
key is the first one to be counted).  
Now, the left hand side is nonnegative and $\phi_1 = \tau^{-1} >0$.
Therefore, $\sum_{i=1}^h \psi_j \geq 0$.

  We now use the claim on the prefix sums of $\psi$ to show that the
  estimation coefficients are nonnegative.
\begin{eqnarray*}
\beta_i & = & \sum_{j=1}^i \psi_j f_{i-j+1} \\
& = &  \sum_{h=1}^i (f_h - f_{h-1})
\sum_{j=1}^{i-h+1} \psi_j\ .
\end{eqnarray*}
We now observe that the right hand side is nonnegative.  This follows from
From monotonicity of $f$ and our claim \eqref{psiprefix:eq} on the nonnegativity of the
$\vecpsi$ prefix sums

\end{proof}
}

\section{The continuous SH spectrum} \label{cont:sec}

  We now present our continuous  SH$_\ell$ sampling
  schemes, which generalizes  SH with weighted
  updates ($\ell = \infty$) \cite{CCD:sigmetrics12}.
The continuous design offers the following advantages over the
  discrete design even when applied to uniform weights.
First,  fixed sample-size sampling no longer requires explicitly maintain a lazy $\seed(x)$
for cached keys as we did in Algorithm \ref{ellaSH:alg}: The lazy
value is implicitly captured by the current threshold $\tau$.  Second,
the estimator 
can be expressed in terms of $f$ and its derivative.
Lastly, the continuous spectrum facilitates multi-objective samples
(Section \ref{MO:sec}).

 Our input is a stream of elements $h=(x,w)$ with key $x$ and a weight
 $w>0$.  
Our element scoring is as follows:
 Each key has a {\em base hash }
$\keybase(x) \sim U[0,1/\ell]$, that is fixed for the computation and
is uniformly distributed in $[0,1/\ell]$:
$\keybase(x) \gets \Hash(x)/\ell$.
An element $h=(x,w)$ is assigned a score 
by first drawing $v \sim \Exp[w]$ and then returning $v$ if $v>1/\ell$
and $\keybase(x)$ otherwise:
{\small
\begin{equation} \label{elementscorecont:eq}
\ElementScore{h} = (v \sim \Exp[w]) \leq 1/\ell\, \,  ? \, 
\, \keybase(x)\,   : \, v \ .
 \end{equation}
}
The random variables
$\Exp[w]$ are independent for different elements
and are also independent of $\keybase(x)$.

  We now consider the distribution of $\seed(x)$ (the minimum 
  element score of stream elements with key $x$).  
\onlyinproc{We use properties of the exponential distribution (see
  TR \cite{freqCap:2015} for details) and show that
$$\seed(x) \sim (v \sim \Exp[w_x]) \leq 1/\ell\, \,  ? \, 
\, U[0,1/\ell]\,   : \, v \ .$$}
\notinproc{
We show that 
$\seed(x) \sim U[0,1/\ell]$ with probability 
$(1-e^{-w_x/\ell})$ and $\seed(x) \sim 1/\ell + \Exp[w_x]$ otherwise:
\begin{lemma}
$$\seed(x) \sim (v \sim \Exp[w_x]) \leq 1/\ell\, \,  ? \, 
\, U[0,1/\ell]\,   : \, v \ .$$
\end{lemma}
\begin{proof}
If at least one of the random variables $\Exp[w(h)]$ for
$h\in x$ is smaller than $1/\ell$, then $\seed(x)=\keybase{x}$.
The distribution of the minimum $\min_{h\in x} \Exp[w(h)]$ is
$\Exp[\sum_{h \in x} w(h)]=\Exp[w_x]$
(the distribution of the minimum of independent 
exponentially distributed random variables with sum of parameters
$w_x$  is exponentially distributed with parameter $w_x$).
So we obtain that if $y\sim \Exp[w_x]$ is such that $y<1/\ell$, which
happens with probability $1-e^{-w_x/\ell}$,  the
seed is $\keybase{x}$.  Otherwise, $\seed(x)=y$.
We now use the memoryless property of the exponential distribution,
which implies that the
conditional distribution of $y-1/\ell$ given that $y>1/\ell$ is
$1/\ell+ \Exp[w_x]$.
So with probability $e^{-w_x/\ell}$, the distribution is $1/\ell + \Exp[w_x]$.
\end{proof}
}
 Note that the element scoring satisfies our requirement (Section
 \ref{escoreprop:sec}) that 
the distribution of $\seed(x)$ depends only on $w_x$.  
 Qualitatively, when $w_x \ll \ell$, the seed is close to 
  exponentially distributed with parameter $w_x$, which is ppswor.  When $w_x \gg 
  \ell$, the seed is uniform, which results in distinct
  sampling. We obtain the property that the
  sampling probability  a key $x$ is roughly proportional to
  $\Cap_\ell(w_x)$, which is needed to approach the ``gold standard'' CV.


\subsection{2-pass estimator}

  Consider 2-pass sampling (Section \ref{2pass:sec}) with
  our element scoring function \eqref{elementscorecont:eq}.
 For estimation, we need to compute the probability
   $\Phi_{\tau,\ell}(w_x) = \Pr[\seed(x) < \tau]$ of a key with weight
  $w_x$ in a sample with parameters $\ell$ and $\tau$.
If $\tau\ell <1$, then a key is included if $\Exp[w_x] < 1/\ell$ and
then $\keybase{x}< \tau$.  These two events are independent and have
joint
probability $(1-e^{-w_x/\ell}) \tau\ell$.  If $\tau\ell \geq 1$ then
a key is included if $\Exp[w_x] < \tau$, which has probability
$(1-e^{-\tau w_x})$.  We can express the combined probability as
\begin{equation} \label{contcumphi}
\Phi_{\tau,\ell}(w_x) \equiv (1-e^{-w_x \max\{1/\ell,\tau\}})\min\{1,\tau\ell\}\
.
\end{equation}
We can then apply inverse probability \eqref{ppsworest} to 
estimate a segment statistics
  $$\hat{Q}(f,H)=\sum_{x\in S\cap H}  \frac{f(w_x)}{\Phi_{\tau,\ell}(w_x)}\ .$$

  We show the following (Proof provided  in \onlyinproc{TR \cite{freqCap:2015}}\notinproc{Appendix~\ref{2passcontashell:sec}}):
\begin{theorem} \label{cv2passSHell:thm}
The CV of estimating $Q(\Cap_T,H)$ from an SH$_\ell$ sample which provides
exact weights $w_x$ ($x\in S$) is
at most 
$$\bigg(\frac{e}{e-1} \frac{\max\{T/\ell,\ell/T\}}{q(k-1)}\bigg)^{0.5}\ .$$
\end{theorem}
When $\ell=T$, we obtain a bound of at most $1.26$ times
the CV bound of $(\frac{\max\{T/\ell,\ell/T\}}{q(k-1)})^{0.5}$ we
can obtain
for samples computed over the aggregated data (Section \ref{prelim:sec}).
When $\ell = \Theta(T)$, the CV is 
$O((q(k-1))^{-0.5})$ and the upper bound degrades smoothly with the
{\em disparity} $\max\{T/\ell,\ell/T\}$ between $\ell$ and $T$.  
Also note
that the increased CV due to the constant
$e/(e-1)$ and the disparity arise from a worst-case analysis and are
not inherent.  

\notinproc{
Figure \ref{incprob:fig}  shows the relative inclusion probabilities as a
function of the weight $w_x$ for SH$_{10}$, and pps and ppswor with
respect to $\Cap_{10}(w_x)$.  The
``gap'' between the ratios for SH$_{10}$ and for the pps/ppswor (which
are realizable on aggregated data) illustrates our loss relative to
the ``gold standard'' CV.  We can see that the gap is larger
for weights that are around the cap parameter of $10$, and maximizes
at the ratio $(1-1/e)$. In this sense,
data with many keys with weight close to the cap thresholds are the
``worst case'' for the variance.
\begin{figure}[htbp]
\centerline{
\resizebox{0.4\textwidth}{!}{\input{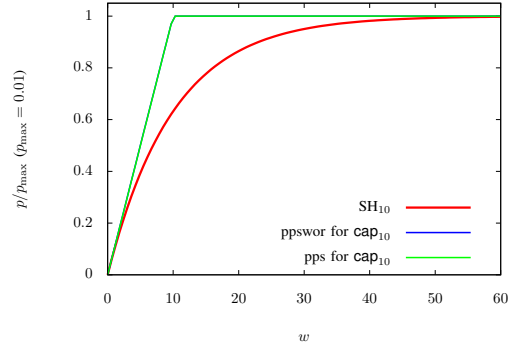}}
}
\caption{(Relative) inclusion probabilities as a
function of the weight $w_x$ for SH$_{10}$, pps, and ppswor, computed with respect
to $\Cap_{10}$.  The $y$ axis shows the inclusion probability as a
function of the maximum inclusion probability. For all schemes we
normalized the threshold to be such that the inclusion
probability maximizes at $0.01$.
 \label{incprob:fig}}
\end{figure}
}

\subsection{1-pass algorithms}

The streaming (1-pass)  algorithms compute a sample $S$  of cached keys and
a value $c_x \leq w_x$ for each $x\in S$.  

Algorithm \ref{cellSHth:alg} performs
fixed threshold sampling.  When processing an element $h=(x,w)$ with a
cached key, we update $c_x \gets c_x+w$.  Otherwise, we compute
the weight $\Delta$ that would be needed for the score to be below
 $\max\{1/\ell,\tau\}$.  If $w\leq \Delta$, we break.
 Otherwise, if $\tau<1/\ell$, we break if $\keybase{x} \geq  \tau$.
  Finally, we initialize a couter $c_x \gets w-\Delta$.
Intuitively, the score is continuously assigned to the
  mass $w_x$.
The value $c_x \leq w_x$ is the weight observed after the point in which the 
score gets below $\tau$.

\begin{algorithm2e}[h]
\caption{Continuous  SH$_\ell$ stream sampling: fixed $\tau$ \label{cellSHth:alg}}
\DontPrintSemicolon 
\SetKwArray{Counters}{Counters}
\SetKwFunction{rand}{rand}
\SetKwFunction{Hash}{Hash}
\SetKwFunction{Return}{return}
\SetKwFunction{ElementScore}{ElementScore}
\KwData{threshold $\tau$, stream of elements $(x,w)$ where
$x\in {\cal X}$ and $w > 0$}
\KwOut{set of pairs $(x,c_x)$ where $x\in {\cal X}$ and $c_x \in
  (0,w_x]$}
$\Counters \gets \emptyset$  \tcp*[h]{Initialize $\Counters$ cache}\; 
\ForEach(\tcp*[h]{Process a stream element}){stream element $(x,w)$}
{
  \If{$x$ is in $\Counters$}{$\Counters{x} \gets \Counters{x} +w$;}
  \Else{
    $\Delta \gets
    -\frac{\ln(1-\rand())}{\max\{\tau,1/\ell\}}$ \tcp*[h]{$\sim \Exp[\max\{\tau,1/\ell\}]$}\;
    \If(\tcp*[h]{initialize counter for $x$}){$\keybase{x} <
      \min\{\tau,1/\ell\}$ and $\Delta < w$}
    {$\Counters{x} \gets w-\Delta$;}
  }
}
\Return{$(x,\Counters{x})$ for $x$ in $\Counters$}
\end{algorithm2e}

Fixed sample size sampling  is provided as Algorithm \ref{cellaSH:alg}. 
To maintain a fixed size sample, the threshold is decreased when there are
$k+1$ cached keys, to the point needed to evict a key. 
The algorithm ``simulates'' the end result of 
working with the lower threshold to begin with.

The eviction step is as follows.  We draw and fix some ``randomization'' and 
compute for each cached key the threshold needed to evict the key. 
The randomization for 
key $x$, in the form of $u_x$ and $r_x$.  We then compute $z_x$ which is 
the maximum threshold value that is needed to evict $x$ with respect 
to that randomization.  
We then take 
the new threshold to be the maximum $z_x$ over keys.  One key (the one 
with maximum $z_x$) is evicted.  For remaining keys, $c_x$ ($\Counters{x}$) is 
updated according to the same $u_x,r_x$.

  We elaborate on how $z_x$ is determined when the current threshold is
  $\tau$.  The key $x$ can be viewed as having 
  a score (computed to the point the key entered the cache) that is at 
  most $\tau$.   We can consider the distribution of the score given
  that it is at most $\tau$: With the randomization, we can take it
as $u_x\tau$.   A necessary requirement for $x$ to be evicted is that the new 
  threshold $\tau^*$ is below $u_x\tau$, so we have  $z_x < u_x\tau$. 
Conditioned on $\tau^*< u_x\tau$,  we can treat 
this as processing  an element with a  new (uncached)  key $x$ and 
weight $c_x$.  We consider the threshold value $\tau^*$ needed for the key 
to enter the cache.  We simply reverse the entry rule:
If $-\ln(1-r_x)/c_x \geq \ell^{-1}$, then the key would enter 
the cache when $\tau^* \geq -\ln(1-r_x)/c_x$. 
If $-\ln(1-r_x)/c_x <  \ell^{-1}$, then the key would enter 
the cache if and only if $\tau^* \geq \keybase{x}$, with count 
$c_x- \ell(-\ln(1-r_x))$.

\begin{algorithm2e}[h]
\caption{Continuous SH$_\ell$ stream sampling: fixed $k$ \label{cellaSH:alg}}
\DontPrintSemicolon 
\SetKwArray{Counters}{Counters}
\SetKwFunction{rand}{rand}
\SetKwFunction{Hash}{Hash}
\SetKwFunction{Return}{return}
\SetKwFunction{Delete}{delete}
\SetKwFunction{ElementScore}{ElementScore}
\KwData{sample size $k$, stream of elements of the form $(x,w)$
  with key $x\in {\cal X}$ and $w > 0$}
\KwOut{$\tau$; set of pairs $(x,c_x)$ where $x\in {\cal X}$ and $c_x \in (0,w_x]$}
$\Counters \gets \emptyset$; $\tau \gets \infty$  \tcp*[h]{Initialize
  cache}\; 
\ForEach(\tcp*[h]{Process element}){stream element $(x,w)$}
{
  \If{$x$ is in $\Counters$}{$\Counters{x} \gets \Counters{x} +w$;}
  \Else{
    $\Delta \gets -\frac{\ln(1-\rand())}{\max\{\ell^{-1},\tau
      \}}$\tcp*[h]{$\sim \Exp[\max\{\ell^{-1},\tau\}]$}\;
    \If (\tcp*[h]{insert $x$}){$\Delta < w$ and ($\tau\ell > 1$ or $\tau\ell \leq 1$ and $\keybase{x}< \tau$)}{
      $\Counters{x} \gets w-\Delta$\;
      \If(\tcp*[h]{Evict a key}){$|\Counters| = k+1$}{
        \If{$\tau\ell > 1$}{
          \ForEach{$x\in \Counters$}{
            $u_x \gets \rand{}$; $r_x \gets \rand{}$ \;
            $z_x \gets \min\{\tau u_x,
            \frac{-\ln(1-r_x)}{\Counters{x}}\}$\tcp*[h]{eviction threshold of $x$}\;
            \If{$z_x \leq \ell^{-1}$}{$z_x \gets \keybase{x}$}
          }
          $y \gets \arg\max_{x\in \Counters} z_x$; Delete
          $y$ from $\Counters$ \tcp*[h]{key to evict} \;
          $\tau^* \gets z_y$\tcp*[h]{new  threshold}\;
          \ForEach(\tcp*[h]{Adjust counters according to $\tau^*$}){$x \in \Counters$}{
            \If{$u_x > \max\{\tau^*,\ell^{-1}\}/\tau$}{$\Counters{x} \overset{-}{\gets}
              \frac{-\ln(1-r_x)}{\max\{\ell^{-1},\tau^*\}}$}
          }
          $\tau \gets \tau^*$;  \Delete $u,r,z,b$  \tcp*[h]{deallocate
            memory}\;
        }
        \Else(\tcp*[h]{$\tau\ell \leq 1$}){
          $y \gets \arg\max_{x\in \Counters} \keybase{x}$; Delete
          $y$ from $\Counters$ \tcp*[h]{evict} \;
          $\tau \gets \keybase{y}$\tcp*[h]{new  threshold}\;
        }
      }
    }
  }
}
\Return{$\tau$;\,  $(x,\Counters{x})$ for $x$ in $\Counters$}
\end{algorithm2e}

 We now express the distribution of $c_x$ ($\Counters{x}$) and verify that it
 satisfies our requirement that for any key $x$,  it 
only depends on $w_x$, $\ell$,  and $\tau$.  
Recall that for fixed-threshold  SH$_\ell$ we use the specified $\tau$ whereas with
fixed-cache size SH$_\ell$, the statement is conditioned on the
randomization on all other keys, which determines $\tau$ when $x\in
S$.
The proof is provided in \onlyinproc{TR \cite{freqCap:2015}}\notinproc{Appendix \ref{contdistc:sec}}.
\begin{theorem} \label{cdensity:thm}
With  fixed-$\tau$ SH$_\ell$ (Algorithm \ref{cellSHth:alg}) and
fixed-$k$ SH$_\ell$ (Algorithm \ref{cellaSH:alg}),
for any key $x$, 
$c_x \sim \max\{0, w_x-\phi\}$,
where $\phi$ has density
$$\phi(y)= \tau \exp(- y \max\{1/\ell,\tau\})\  $$
in the interval $y\in [0,w_x]$.
\end{theorem}

\notinproc{
\subsubsection*{Optimization comment: Batch evictions}
   Each decrease of the threshold involves scanning all keys in the
   cache.  The expected total number of evictions, however, is at most $k\ln
   m$, where $m$ is the number of elements.
To reduce amortized eviction cost, we can use a slight modification of the
   algorithm which evicts $\delta$ keys when the cache is full, where $\delta$ is a
   fraction of $k$.  The
modification uses 
the $\delta$th largest $z_x$ instead of the maximum one as the new $\tau^*$.  All keys
$y$ with $z_y \geq \tau$ are then evicted.  The new threshold is
$\tau^*$.  
The computation of the
estimators, which we present next, is with respect to the current
threshold and is the same with the batched evictions or one at a time eviction.
}


\subsection{Estimators for Continuous SH$_\ell$}

\ignore{

We can also consider a distribution $\vecm$ over key sizes and a respected
expected distribution $\veco$ over the counted sizes in the sample.

  We have the relation by the following integral transform:

\begin{equation}
o(c) = \int_c^\infty m(w) \phi(w-c) dw 
\end{equation}

  We are interested in the inverse transform, stated as a function
  $\psi$, which expresses $\vecm$
  in terms of the expected $\veco$:
\begin{equation}
m(w) = \int_w^\infty o(c) \psi(c-w) dc 
\end{equation}

Combining, we obtain the relation
\begin{equation}
m(w) = \int_w^\infty  \psi(c-w) \int_c^\infty m(x) \phi(x-c) dx   dc 
\end{equation}

  From $\psi$ we can obtain estimation coefficient function which can
  be used in an estimator \eqref{coefform:eq}
\begin{equation}
\beta(x) = \int_0^x \psi(y) f(x-y) dy
\end{equation}

}

  We seek an unbiased and nonnegative estimator in a coefficient form, that
  is, a function $\beta^{(f,\tau,\ell)}(c)$ defined for any   $c>0$
  and we use the estimator 
\begin{equation}  \label{contest:eq}
\hat{Q}(f,H) = \sum_{x\in H\cap S}
  \beta^{(f,\tau,\ell)}(c_x)\ .
\end{equation}
\begin{theorem} \label{estcoeffcont:thm}
For any continuous $f$ that is differentiable almost everywhere, the
estimator that uses
\begin{equation} \label{contbeta:eq}
\beta^{(f,\tau,\ell)}(c) \equiv  f(c)/\min\{1, \ell\tau\} + f'(c)/\tau
\ 
\end{equation}
is unbiased.
\end{theorem}
\begin{proof}  
We separately
 treat the cases where $\tau\ell<1$ and $\tau\ell>1$.
We first show that when $\tau\ell >1$,
 $\beta(c) = f(c) + f'(c)/\tau $ are unbiased
  estimation coefficients.

For a key of size
  $w$, we have density $\tau e^{-\tau x}$ to have count of $w-x \in
  (0,w)$ (otherwise the key has count $0$ and the estimate is $0$).
 We can write $$\beta(y) =  (f(y) e^{\tau y})' e^{-\tau y}
 \tau^{-1}\ .$$
  Consider a key of size $w$.  Its expected contribution to the
  estimate is
{\small
\begin{eqnarray*}
\lefteqn{\int_0^w \tau e^{-\tau x} \beta(w-x) dx }\\
&=& \int_0^w \tau e^{-\tau x} (f(w-x) e^{-\tau(w-x)})'
e^{-\tau(w-x)} \tau^{-1} dx \\
&=& e^{-\tau w} \int_0^w (f(w-x)e^{-\tau(w-x)})' dx \\
&=& e^{-\tau w} f(w) e^{-\tau w} = f(w)
\end{eqnarray*}
}

  We now consider the case where $\tau < 1/\ell$, showing that
 $$\beta(c) = f(c)/(\ell\tau) + f'(c)/\tau $$ are unbiased
  estimation coefficients.
For a key with weight
  $w$, we have density $\tau e^{-x/\ell}$ to have count of $w-x \in
  (0,w)$.
We write
$$\beta(y) = (f(y) e^{y/\ell})' e^{-y/\ell} \tau^{-1}\ .$$
\begin{eqnarray*}
\lefteqn{\int_0^w \tau e^{-x/\ell} \beta(w-x) dx} \\
&=& \int_0^w \tau e^{-x/\ell} (f(w-x) e^{(w-x)/\ell})' e^{-(w-x)/\ell}
\tau^{-1}\\
&=& e^{-w/\ell} \int_0^w (f(w-x)e^{(w-x)/\ell})'dx = f(w)\ .
\end{eqnarray*}
\end{proof}
Note that any continuous monotone function, including the
$\Cap_T$ functions, is differentiable almost everywhere and hence satisfies
the requirements of the theorem.

 We upper bound the CV of the streaming fixed-$k$ SH$_\ell$ estimator (Proof is
in \onlyinproc{TR \cite{freqCap:2015}}\notinproc{ Appendix \ref{1passcontashell:sec}}):
\begin{theorem} \label{cv1passSHell:thm}
The CV of estimating $Q(\Cap_T,H)$ from an SH$_\ell$ sample is upper bounded 
by 
\begin{eqnarray*}
\lefteqn{\frac{e}{e-1}\max\{1,\frac{\ell}{T}\}\bigg( \frac{\ell}{T}(1-e^{-T/\ell}) +
\frac{T}{\ell} \bigg) }\\
&\leq& \bigg(\frac{\frac{e}{e-1}
  (1+\max\{\ell/T,T/\ell\})}{q(k-1)}\bigg)^{0.5}\ .
\end{eqnarray*}
\end{theorem}
In particular, when
$\ell = \Theta(T)$, the CV  is $O(q(k-1)^{-0.5})$, and when 
$\ell=T$, the CV is at
most $$\bigg(\frac{2e-1}{e-1}\frac{1}{q(k-1)}\bigg)^{0.5} \approx 1.6 (q(k-1)^{-0.5}\ .$$

\section{Multi-objective samples}   \label{MO:sec}

  We established that from a fixed-$k$ SH$_\ell$ sample we can estimate well 
$\Cap_T$ statistics when $T=\Theta(\ell)$.
  This means that if we are interested in estimates with statistical
  guarantees for cap values $T=[a,b]$, it suffices to use SH$_\ell$ samples with
parameters  $\ell_i=2^i a$ for $i\leq \lceil \log(b/a)\rceil$.  
To process a query for a $\Cap_T$ statistics, we can  use the SH$_\ell$ sample with
 $\ell$ that is closest (within a factor of $\sqrt{2}$) from $T$.
In particular,
  to estimate all cap statistics, it suffices to use 
$\lceil \log(\max_x w_x/\min_x   w_x))\rceil$ samples.

 We now improve over this basic approach \onlyinproc{(see TR \cite{freqCap:2015} for
   details) }by instead of working with a
 set $\{S_\ell\}$ of samples with respective caps $\ell\in L$, we work with a single sample
$S_L=\bigcup_{\ell\in L} S_\ell$.  The improvement has several components:
Sample coordination, which ensures that samples with closer $\ell$
are more similar so that 
$|S_L| \ll k|L|$,  using estimators that benefit from the combined sample, and
sampling algorithms that  use state that is proportional to $|S_L|$.
\notinproc{
\subsection{Sample coordination}
We {\em coordinate} the samples for different
$\ell$ \cite{BrEaJo:1972,multiw:VLDB2009} 
by using the same ``randomization.''  In our context, the randomization
of each key constitutes of two independent random variables $\Hash(x)
\sim U[0,1]$ and 
$y_x \sim  \Exp[w_x]$, which is the minimum over elements with key $x$ of the $\Exp[w]$
component used for scoring elements $(x,w)$.
With coordination, we can express the seed of $x$ as a function of $\ell$ as:
$$\seed_\ell(x) = y_x\leq 1/\ell\, ?\, \Hash{x}/\ell \, : \,  y_x\ \ .$$
The sample $S_\ell$ includes the $k$ keys with smallest
$\seed_\ell$ values and its threshold $\tau_\ell$ is the $(k+1)$st
smallest $\seed_\ell$ value (or $+\infty$ if there are fewer than
$k+1$ active keys).

 Surprisingly perhaps, we show that the expected number of distinct
 keys in $S_L$ for $L=(0,\infty)$  when the samples $S_\ell$ are coordinated is
at most $k\ln n$, where  $n$ is the number of active keys in the data set.
  In particular, this upper bounds $|S_L|$ for {\em any} set of cap 
parameters $L$.
\begin{lemma}
Let $\{S_\ell\}$ for $\ell\in L=(0,\infty)$ be coordinated  fixed-$k$ SH$_\ell$
samples.
Then  $\E[|S_L|] = \E[|\bigcup_{\ell>0} S_\ell|] < k\ln n$.  Moreover,
for $a>1$, the probability of $|S_L| > a k\ln n$ decreases exponentially with $a$.
\end{lemma}
\begin{proof}
Consider an order of all keys by increasing  $y_x$.  Any sample
$S_\ell$ must have the form of the keys with $k$ smallest $\Hash{x}$
values in a prefix of this order.  
We now consider the $i$th key in this order and
the probability that it qualifies for some sample, which is the 
probability that $\Hash{x}$ is among the $k$ smallest in the prefix of
the first $i$ keys.
Since the random variables $\Hash{x}$ are
independent of $y_x$ and unrelated to the order, this is exactly
the probability that the key  is in one of the first $k$ positions
in a random permutation of size $i$,  which is $\min\{1, k/i\}$.
Summing over all $i$ we obtain the claim.  Concentration follows from
Chernoff bounds.
\end{proof}
A corollary of the proof is that the property $x\in S_\ell$ holds for
a contiguous interval of $\ell$ values that generally has the form
$(1/y_z,1/y_x]$ for some key $z$. If $y_x$ is amongst the $k$
smallest among $\{ y_z \mid z\in {\cal X}\}$ then the interval has the form
$(1/y_z,+\infty)$ and if $\Hash{x}$ is amongst the $k$ smallest in
$\{ \Hash{z} \mid z\in {\cal X}\}$ then
the interval is $(0,1/y_x]$.

\subsection{Estimation}
We now consider estimators that leverage all
the sampled keys $x\in S_L$ \cite{multiw:VLDB2009}.  This allows us to obtain tighter
estimates than when using any one sample $S_\ell$.
We consider 2-pass sampling, so that $w_x$ is available
for sampled keys, and compute for each key $x$ a probability
$\Phi(w_x)$ that it is included in {\em at least one} of $S_\ell$ for
$\ell\in L$, when fixing
the randomization on ${\cal X}\setminus\{x\}$.  

Once we have the probabilities $\Phi(w_x)$ we can apply an inverse 
probability estimate $\hat{Q}(f,H) = \sum_{x\in H} f(w_x)/\Phi(w_x)$. 
Since $\Phi(w_x)$ is at least as large as the inclusion probability of 
$x$ in any individual $S_\ell$,  the variance of the estimate for any 
query is at most that obtained by using any single sample.

 We now elaborate on computing the probabilities $\Phi$.
This probability can be
computed from $w_x$ and the set of pairs $\{\ell,\tau^{-x}_\ell\}$ for
$\ell\in L$.  Here,  we define
$\tau^{-x}_\ell$ to be the k$^{\text{th}}$ smallest $\seed_\ell(z)$ for $z\in {\cal X}\setminus
\{x\}$.  When $x\in S_\ell$, this is the threshold $\tau_\ell$
and otherwise, it is $\max_{z\in S_\ell}\seed_\ell(z)$.
\begin{lemma}
$$\Phi(x) = \Pr_{y\sim\Exp[w_x], h\sim U[0,1]} [\exists \ell\in L,\
C(\ell,x)]\ ,$$
where
$C(\ell,x)$ is the condition
$y < \max\{\tau^{-x}_\ell ,1/\ell\}$ and $h  < \ell \tau^{-x}_\ell$.
\end{lemma}
\begin{proof}
Using the independent random variables $y_x \sim \Exp[w_x]$ and $\Hash{x}$, 
the condition for inclusion of $x$ in $S_\ell$ is that
$$y_x < \max\{\tau^{-x}_\ell ,1/\ell\}\,  \text{ and } \, \Hash(x)  <
\ell \tau^{-x}_\ell\ .$$
The condition for inclusion in $S_L$ is that $(y_x,\Hash{x})$ satisfy
the condition for at least one $\ell\in L$.
\end{proof}
When working with $L$ that contains a contiguous interval, such as $L=(0,\infty)$, we can express the $k$th and
$(k+1)$st smallest values in $\seed_\ell(z)$ as a function of $1/\ell$
as a piecewise linear function with at most $k\ln n$ pieces (in expectation).
This allows us to compute $\Phi(w_x)$ for all keys in $x\in S_L$.

\subsection{Sampling algorithm}
 We can engineer the sampling algorithm so that it
maintains state that is proportional to the number of
 distinct keys in $|S_L|$.   Let $\ell_i$ be our list of cap
 parameters in decreasing order.  The algorithm maintains 
for each $i$, the $k+1$ keys with smallest $\Hash{x}$ amongst those
with $y_x < 1/\ell_i$.  If for the highest $\ell$ values in $L$ there
are fewer than $k+1$ keys with $y_x<1/\ell$, we include the $k+1$ keys with smallest $y_x$.

}

\begin{figure*}[t]
{\tiny 
\notinproc{
\parbox{.45\linewidth}{\centering 
  discrete $ k =  100 $, $\alpha =  1.2$,  $m =  10^5$, $rep =  200$,   relerr  1-pass  \\
\begin{tabular}{r | r r r r r r r r}
\hline 
 $\ell$, $T$    &  1  &  5  &  20  &  50  &  100  &  500  &  1000  &  10000  \\ 
\hline 
1    &   0.079    &   0.090    &   0.147    &   0.221    &   0.289    &   0.491    &   0.580    &   1.080    \\ 
5    & {\bf  0.076}   & {\bf  0.075}   &   0.085    &   0.109    &   0.137    &   0.253    &   0.350    &   0.754    \\ 
20    &   0.109    &   0.088    &   0.079    &   0.085    &   0.092    &   0.149    &   0.190    &   0.439    \\ 
50    &   0.105    &   0.083    & {\bf  0.077}   &   0.078    &   0.085    &   0.131    &   0.166    &   0.346    \\ 
100    &   0.115    &   0.103    &   0.083    & {\bf  0.078}   & {\bf  0.078}   &   0.087    &   0.103    &   0.260    \\ 
500    &   0.135    &   0.110    &   0.099    &   0.090    &   0.087    &   0.082    &   0.081    &   0.120    \\ 
1000    &   0.133    &   0.123    &   0.110    &   0.100    &   0.094    &   0.079    &   0.074    &   0.072    \\ 
10000    &   0.142    &   0.118    &   0.103    &   0.087    &   0.080    & {\bf  0.068}   & {\bf  0.061}   & {\bf  0.045}   \\ 
\end{tabular} 
}
\parbox{.45\linewidth}{\centering 
  discrete $ k =  100 $, $\alpha =  1.2 $,  $m =  10^5$, $rep =  200 $,   rel err  2-pass  \\
\begin{tabular}{||r | r r r r r r r r}
\hline 
 $\ell$, $T$    &  1  &  5  &  20  &  50  &  100  &  500  &  1000  &  10000  \\ 
\hline 
1    &   0.079    &   0.090    &   0.147    &   0.221    &   0.289    &   0.491    &   0.580    &   1.080    \\ 
5    & {\bf  0.075}   & {\bf  0.075}   &   0.085    &   0.109    &   0.138    &   0.253    &   0.350    &   0.754    \\ 
20    &   0.099    &   0.086    &   0.079    &   0.084    &   0.092    &   0.150    &   0.190    &   0.439    \\ 
50    &   0.100    &   0.084    & {\bf  0.077}   & {\bf  0.077}   &   0.083    &   0.128    &   0.164    &   0.345    \\ 
100    &   0.116    &   0.099    &   0.084    &   0.079    & {\bf  0.077}   &   0.086    &   0.102    &   0.260    \\ 
500    &   0.124    &   0.106    &   0.098    &   0.092    &   0.088    &   0.081    &   0.080    &   0.118    \\ 
1000    &   0.125    &   0.115    &   0.109    &   0.102    &   0.095    &   0.078    &   0.073    &   0.070    \\ 
10000    &   0.133    &   0.111    &   0.096    &   0.086    &   0.080    & {\bf  0.066}   & {\bf  0.061}   & {\bf  0.044}   \\ 
\end{tabular} 
}
}

\parbox{.45\linewidth}{\centering 
  discrete $ k =  100 $, $\alpha =  1.2 $,  $m =  10^5$, $rep =  200$,   NRMSE 1-pass  \\
\begin{tabular}{r | r r r r r r r r}
\hline 
 $\ell$, $T$    &  1  &  5  &  20  &  50  &  100  &  500  &  1000  &  10000  \\ 
\hline 
1    &   0.098    &   0.115    &   0.185    &   0.279    &   0.374    &   0.658    &   0.862    &   3.016    \\ 
5    & {\bf  0.094}   & {\bf  0.093}   &   0.112    &   0.144    &   0.184    &   0.332    &   0.449    &   1.316    \\ 
20    &   0.133    &   0.111    &   0.102    &   0.109    &   0.122    &   0.199    &   0.254    &   0.615    \\ 
50    &   0.138    &   0.108    & {\bf  0.098}   &   0.101    &   0.107    &   0.163    &   0.207    &   0.419    \\ 
100    &   0.146    &   0.125    &   0.104    & {\bf  0.099}   & {\bf  0.099}   &   0.111    &   0.133    &   0.311    \\ 
500    &   0.171    &   0.135    &   0.123    &   0.112    &   0.110    &   0.102    &   0.101    &   0.149    \\ 
1000    &   0.174    &   0.156    &   0.141    &   0.125    &   0.118    &   0.100    &   0.094    &   0.090    \\ 
10000    &   0.178    &   0.148    &   0.128    &   0.110    &   0.102    & {\bf  0.083}   & {\bf  0.076}   & {\bf  0.056}   \\ 
\end{tabular} 
}
\parbox{.45\linewidth}{\centering 
  discrete $ k =  100 $, $\alpha =  1.2 $,  $m =  10^5$, $rep =  200$,   NRMSE 2-pass  \\
\begin{tabular}{||r | r r r r r r r r}
\hline 
 $\ell$, $T$    &  1  &  5  &  20  &  50  &  100  &  500  &  1000  &  10000  \\ 
\hline 
1    &   0.098    &   0.115    &   0.185    &   0.279    &   0.374    &   0.658    &   0.862    &   3.016    \\ 
5    & {\bf  0.094}   & {\bf  0.093}   &   0.110    &   0.143    &   0.183    &   0.333    &   0.449    &   1.316    \\ 
20    &   0.123    &   0.109    &   0.101    &   0.108    &   0.120    &   0.199    &   0.254    &   0.614    \\ 
50    &   0.131    &   0.108    & {\bf  0.100}   & {\bf  0.099}   &   0.105    &   0.161    &   0.205    &   0.417    \\ 
100    &   0.144    &   0.122    &   0.105    &   0.099    & {\bf  0.097}   &   0.109    &   0.131    &   0.310    \\ 
500    &   0.156    &   0.130    &   0.120    &   0.114    &   0.110    &   0.101    &   0.099    &   0.147    \\ 
1000    &   0.161    &   0.148    &   0.137    &   0.126    &   0.118    &   0.099    &   0.092    &   0.088    \\ 
10000    &   0.165    &   0.140    &   0.120    &   0.107    &   0.099    & {\bf  0.082}   & {\bf  0.075}   & {\bf  0.054}   \\ 
\end{tabular} 
} 

\notinproc{
\parbox{.45\linewidth}{\centering 
  discrete $ k =  100 $, $\alpha =  1.5 $,  $m =  10^5 $, $rep =  500 $,   relerr  1-pass  \\
\begin{tabular}{r | r r r r r r r r}
\hline 
 $\ell$, $T$    &  1  &  5  &  20  &  50  &  100  &  500  &  1000  &
 10000  \\ 
\hline 
1    & {\bf  0.081}   &   0.105    &   0.151    &   0.205    &   0.256    &   0.439    &   0.556    &   0.958    \\ 
5    &   0.091    & {\bf  0.075}   &   0.091    &   0.114    &   0.138    &   0.243    &   0.309    &   0.687    \\ 
20    &   0.122    &   0.089    & {\bf  0.074}   &   0.080    &   0.091    &   0.146    &   0.188    &   0.419    \\ 
50    &   0.145    &   0.109    &   0.087    &   0.078    &   0.078    &   0.101    &   0.125    &   0.290    \\ 
100    &   0.150    &   0.111    &   0.083    & {\bf  0.070}   & {\bf  0.066}   &   0.077    &   0.093    &   0.199    \\ 
500    &   0.176    &   0.123    &   0.098    &   0.080    &   0.069    &   0.051    &   0.045    &   0.043    \\ 
1000    &   0.175    &   0.125    &   0.093    &   0.078    &   0.068    & {\bf  0.047}   &   0.038    &   0.022    \\ 
10000    &   0.175    &   0.134    &   0.100    &   0.081    &   0.071    &   0.047    & {\bf  0.037}   & {\bf  0.019}   \\ 
\end{tabular} 
}
\parbox{.45\linewidth}{\centering 
  discrete $ k =  100 $, $\alpha =  1.5 $,  $m =  10^5 $, $rep =  500 $,   rel err  2-pass  \\
\begin{tabular}{||r | r r r r r r r r}
\hline 
 $\ell$, $T$    &  1  &  5  &  20  &  50  &  100  &  500  &  1000  &
 10000  \\ 
\hline 
1    & {\bf  0.081}   &   0.105    &   0.151    &   0.205    &   0.256    &   0.439    &   0.556    &   0.958    \\ 
5    &   0.087    & {\bf  0.075}   &   0.091    &   0.114    &   0.138    &   0.244    &   0.309    &   0.687    \\ 
20    &   0.110    &   0.087    & {\bf  0.074}   &   0.079    &   0.090    &   0.144    &   0.187    &   0.419    \\ 
50    &   0.129    &   0.100    &   0.085    &   0.078    &   0.078    &   0.101    &   0.125    &   0.290    \\ 
100    &   0.137    &   0.103    &   0.079    & {\bf  0.069}   & {\bf  0.065}   &   0.078    &   0.093    &   0.199    \\ 
500    &   0.159    &   0.117    &   0.092    &   0.077    &   0.068    &   0.049    &   0.043    &   0.042    \\ 
1000    &   0.152    &   0.112    &   0.090    &   0.075    &   0.066    & {\bf  0.044}   &   0.036    &   0.019    \\ 
10000    &   0.163    &   0.121    &   0.094    &   0.079    &   0.070    &   0.045    & {\bf  0.036}   & {\bf  0.018}   \\ 
\end{tabular} 
}

\parbox{.45\linewidth}{\centering 
  discrete $ k =  100 $, $\alpha =  1.5 $,  $m =  10^5 $, $rep =  500 $,   NRMSE 1-pass  \\
\begin{tabular}{r | r r r r r r r r}
\hline 
 $\ell$, $T$    &  1  &  5  &  20  &  50  &  100  &  500  &  1000  &
 10000  \\ 
\hline 
1    & {\bf  0.102}   &   0.133    &   0.193    &   0.267    &   0.330    &   0.556    &   0.700    &   1.448    \\ 
5    &   0.114    & {\bf  0.097}   &   0.118    &   0.148    &   0.181    &   0.312    &   0.396    &   0.862    \\ 
20    &   0.153    &   0.111    & {\bf  0.094}   &   0.101    &   0.115    &   0.183    &   0.230    &   0.508    \\ 
50    &   0.184    &   0.137    &   0.112    &   0.100    &   0.099    &   0.128    &   0.156    &   0.353    \\ 
100    &   0.192    &   0.141    &   0.106    & {\bf  0.091}   & {\bf  0.084}   &   0.097    &   0.114    &   0.240    \\ 
500    &   0.225    &   0.157    &   0.122    &   0.101    &   0.087    &   0.064    &   0.057    &   0.058    \\ 
1000    &   0.221    &   0.160    &   0.119    &   0.098    &   0.086    & {\bf  0.060}   & {\bf  0.049}   &   0.029    \\ 
10000    &   0.223    &   0.171    &   0.127    &   0.103    &   0.091    &   0.061    &   0.049    & {\bf  0.025}   \\ 
\end{tabular} 
}
\parbox{.45\linewidth}{\centering 
  discrete $ k =  100 $, $\alpha =  1.5 $,  $m =  10^5 $, $rep =  500 $,   NRMSE 2-pass  \\
\begin{tabular}{||r | r r r r r r r r}
\hline 
 $\ell$, $T$    &  1  &  5  &  20  &  50  &  100  &  500  &  1000  &
 10000  \\ 
\hline 
1    & {\bf  0.102}   &   0.133    &   0.193    &   0.267    &   0.330    &   0.556    &   0.700    &   1.448    \\ 
5    &   0.110    & {\bf  0.096}   &   0.116    &   0.147    &   0.180    &   0.312    &   0.397    &   0.862    \\ 
20    &   0.139    &   0.109    & {\bf  0.094}   &   0.099    &   0.112    &   0.182    &   0.230    &   0.508    \\ 
50    &   0.163    &   0.129    &   0.108    &   0.100    &   0.099    &   0.127    &   0.156    &   0.353    \\ 
100    &   0.175    &   0.133    &   0.102    & {\bf  0.088}   & {\bf  0.083}   &   0.097    &   0.115    &   0.240    \\ 
500    &   0.198    &   0.148    &   0.115    &   0.097    &   0.086    &   0.062    &   0.054    &   0.057    \\ 
1000    &   0.191    &   0.144    &   0.115    &   0.097    &   0.084    & {\bf  0.056}   & {\bf  0.045}   &   0.024    \\ 
10000    &   0.207    &   0.154    &   0.120    &   0.102    &   0.089    &   0.058    &   0.047    & {\bf  0.023}   \\ 
\end{tabular} 
}

\parbox{.45\linewidth}{\centering 
  discrete $ k =  50 $, $\alpha =  1.8 $,  $m =  10^5 $, $rep =  500 $,   relerr  1-pass  \\
\begin{tabular}{r | r r r r r r r r}
\hline 
 $\ell$, $T$    &  1  &  5  &  20  &  50  &  100  &  500  &  1000  &  10000  \\ 
\hline 
1    & {\bf  0.104}   &   0.130    &   0.189    &   0.245    &   0.301    &   0.490    &   0.585    &   1.022    \\ 
5    &   0.136    & {\bf  0.111}   &   0.123    &   0.148    &   0.175    &   0.279    &   0.346    &   0.713    \\ 
20    &   0.188    &   0.123    & {\bf  0.101}   &   0.101    &   0.114    &   0.180    &   0.222    &   0.430    \\ 
50    &   0.238    &   0.156    &   0.116    & {\bf  0.099}   &   0.098    &   0.130    &   0.156    &   0.315    \\ 
100    &   0.260    &   0.164    &   0.122    &   0.100    &   0.091    &   0.097    &   0.113    &   0.228    \\ 
500    &   0.318    &   0.199    &   0.143    &   0.114    &   0.095    &   0.061    &   0.053    &   0.045    \\ 
1000    &   0.315    &   0.203    &   0.131    &   0.108    &   0.092    &   0.054    &   0.044    &   0.019    \\ 
10000    &   0.324    &   0.215    &   0.144    &   0.116    & {\bf  0.089}   & {\bf  0.051}   & {\bf  0.039}   & {\bf  0.016}   \
\end{tabular} 
}
\parbox{.45\linewidth}{\centering 
  discrete $ k =  50 $, $\alpha =  1.8 $,  $m =  10^5 $, $rep =  500 $,   rel err  2-pass  \\
\begin{tabular}{|| r | r r r r r r r r}
\hline 
 $\ell$, $T$    &  1  &  5  &  20  &  50  &  100  &  500  &  1000  &  10000  \\ 
\hline 
1    & {\bf  0.104}   &   0.130    &   0.189    &   0.245    &   0.301    &   0.490    &   0.585    &   1.022    \\ 
5    &   0.133    & {\bf  0.110}   &   0.120    &   0.147    &   0.175    &   0.279    &   0.346    &   0.713    \\ 
20    &   0.170    &   0.119    & {\bf  0.098}   &   0.100    &   0.113    &   0.180    &   0.223    &   0.430    \\ 
50    &   0.214    &   0.152    &   0.115    &   0.099    &   0.096    &   0.128    &   0.154    &   0.315    \\ 
100    &   0.218    &   0.148    &   0.117    & {\bf  0.099}   &   0.090    &   0.097    &   0.112    &   0.228    \\ 
500    &   0.253    &   0.179    &   0.135    &   0.111    &   0.094    &   0.058    &   0.049    &   0.042    \\ 
1000    &   0.269    &   0.182    &   0.129    &   0.105    & {\bf  0.088}   &   0.051    &   0.039    &   0.016    \\ 
10000    &   0.282    &   0.188    &   0.136    &   0.107    &   0.088    & {\bf  0.050}   & {\bf  0.037}   & {\bf  0.014}   \\ 
\end{tabular} 
}

\parbox{.45\linewidth}{\centering 
  discrete $ k =  50 $, $\alpha =  1.8 $,  $m =  10^5 $, $rep =  500 $,   NRMSE 1-pass  \\
\begin{tabular}{r | r r r r r r r r}
\hline 
 $\ell$, $T$    &  1  &  5  &  20  &  50  &  100  &  500  &  1000  &  10000  \\ 
\hline 
1    & {\bf  0.133}   &   0.169    &   0.247    &   0.321    &   0.396    &   0.630    &   0.753    &   1.377    \\ 
5    &   0.172    & {\bf  0.143}   &   0.162    &   0.194    &   0.230    &   0.365    &   0.445    &   0.866    \\ 
20    &   0.234    &   0.156    & {\bf  0.128}   &   0.129    &   0.143    &   0.224    &   0.275    &   0.531    \\ 
50    &   0.302    &   0.191    &   0.143    &   0.126    &   0.126    &   0.165    &   0.198    &   0.385    \\ 
100    &   0.327    &   0.206    &   0.154    & {\bf  0.126}   &   0.113    &   0.123    &   0.142    &   0.276    \\ 
500    &   0.397    &   0.252    &   0.181    &   0.150    &   0.125    &   0.080    &   0.069    &   0.065    \\ 
1000    &   0.404    &   0.258    &   0.168    &   0.137    &   0.116    &   0.069    &   0.056    &   0.025    \\ 
10000    &   0.416    &   0.272    &   0.181    &   0.145    & {\bf  0.112}   & {\bf  0.064}   & {\bf  0.049}   & {\bf  0.020}   \
\end{tabular} 
}
\parbox{.45\linewidth}{\centering 
  discrete $ k =  50 $, $\alpha =  1.8 $,  $m =  10^5 $, $rep =  500 $,   NRMSE 2-pass  \\
\begin{tabular}{||r | r r r r r r r r}
\hline 
 $\ell$, $T$    &  1  &  5  &  20  &  50  &  100  &  500  &  1000  &  10000  \\ 
\hline 
1    & {\bf  0.133}   &   0.169    &   0.247    &   0.321    &   0.396    &   0.630    &   0.753    &   1.377    \\ 
5    &   0.167    & {\bf  0.140}   &   0.160    &   0.194    &   0.230    &   0.366    &   0.445    &   0.866    \\ 
20    &   0.214    &   0.152    & {\bf  0.126}   &   0.127    &   0.143    &   0.224    &   0.276    &   0.532    \\ 
50    &   0.270    &   0.186    &   0.141    &   0.125    &   0.124    &   0.163    &   0.197    &   0.385    \\ 
100    &   0.276    &   0.187    &   0.143    & {\bf  0.123}   &   0.112    &   0.121    &   0.140    &   0.276    \\ 
500    &   0.316    &   0.225    &   0.172    &   0.144    &   0.123    &   0.076    &   0.064    &   0.064    \\ 
1000    &   0.347    &   0.231    &   0.165    &   0.134    &   0.112    &   0.065    &   0.050    &   0.021    \\ 
10000    &   0.365    &   0.235    &   0.170    &   0.133    & {\bf  0.108}   & {\bf  0.061}   & {\bf  0.045}   & {\bf  0.018}   \\
\end{tabular} 
}

\parbox{.45\linewidth}{\centering 
  discrete $ k =  50 $, $\alpha =  2 $,  $m =  10^5 $, $rep =  500 $,   relerr  1-pass  \\
\begin{tabular}{r | r r r r r r r r}
\hline 
 $\ell$, $T$    &  1  &  5  &  20  &  50  &  100  &  500  &  1000  &  10000  \\ 
\hline 
1    & {\bf  0.116}   &   0.136    &   0.184    &   0.220    &   0.261    &   0.387    &   0.466    &   0.862    \\ 
5    &   0.136    & {\bf  0.105}   &   0.114    &   0.132    &   0.159    &   0.241    &   0.294    &   0.517    \\ 
20    &   0.191    &   0.123    &   0.098    &   0.097    &   0.107    &   0.154    &   0.184    &   0.338    \\ 
50    &   0.225    &   0.144    & {\bf  0.094}   & {\bf  0.083}   &   0.082    &   0.107    &   0.127    &   0.227    \\ 
100    &   0.265    &   0.166    &   0.112    &   0.090    &   0.078    &   0.075    &   0.083    &   0.149    \\ 
500    &   0.314    &   0.171    &   0.114    &   0.086    &   0.068    &   0.036    &   0.027    &   0.011    \\ 
1000    &   0.303    &   0.192    &   0.120    &   0.086    &   0.068    &   0.037    &   0.028    &   0.011    \\ 
10000    &   0.315    &   0.189    &   0.120    &   0.085    & {\bf  0.065}   & {\bf  0.033}   & {\bf  0.025}   & {\bf  0.009}   \\
\end{tabular} 
}
\parbox{.45\linewidth}{\centering 
  discrete $ k =  50 $, $\alpha =  2 $,  $m =  10^5$, $rep =  500 $,   rel err  2-pass  \\
\begin{tabular}{||r | r r r r r r r r}
\hline 
 $\ell$, $T$    &  1  &  5  &  20  &  50  &  100  &  500  &  1000  &  10000  \\ 
\hline 
1    & {\bf  0.116}   &   0.136    &   0.184    &   0.220    &   0.261    &   0.387    &   0.466    &   0.862    \\ 
5    &   0.131    & {\bf  0.103}   &   0.111    &   0.131    &   0.158    &   0.241    &   0.294    &   0.517    \\ 
20    &   0.164    &   0.119    &   0.097    &   0.097    &   0.107    &   0.155    &   0.184    &   0.338    \\ 
50    &   0.192    &   0.130    & {\bf  0.094}   & {\bf  0.080}   &   0.079    &   0.106    &   0.126    &   0.227    \\ 
100    &   0.231    &   0.152    &   0.109    &   0.087    &   0.076    &   0.072    &   0.081    &   0.148    \\ 
500    &   0.253    &   0.162    &   0.112    &   0.084    &   0.065    &   0.033    &   0.024    &   0.010    \\ 
1000    &   0.262    &   0.172    &   0.113    &   0.085    &   0.066    &   0.032    &   0.023    &   0.008    \\ 
10000    &   0.259    &   0.168    &   0.114    &   0.082    & {\bf  0.063}   & {\bf  0.031}   & {\bf  0.022}   & {\bf  0.008} \\
\end{tabular} 
}
}

\parbox{.45\linewidth}{\centering 
  discrete $ k =  50 $, $\alpha =  2 $,  $m =  10^5$, $rep =  500 $,   NRMSE 1-pass  \\
\begin{tabular}{r | r r r r r r r r}
\hline 
 $\ell$, $T$    &  1  &  5  &  20  &  50  &  100  &  500  &  1000  &  10000  \\ 
\hline 
1    & {\bf  0.145}   &   0.172    &   0.235    &   0.290    &   0.345    &   0.505    &   0.601    &   1.063    \\ 
5    &   0.174    & {\bf  0.134}   &   0.147    &   0.170    &   0.202    &   0.311    &   0.370    &   0.636    \\ 
20    &   0.243    &   0.153    &   0.123    &   0.126    &   0.138    &   0.196    &   0.232    &   0.421    \\ 
50    &   0.280    &   0.181    & {\bf  0.120}   & {\bf  0.106}   &   0.104    &   0.134    &   0.160    &   0.282    \\ 
100    &   0.343    &   0.211    &   0.146    &   0.116    &   0.099    &   0.097    &   0.107    &   0.185    \\ 
500    &   0.397    &   0.222    &   0.141    &   0.107    &   0.085    &   0.046    &   0.035    &   0.018    \\ 
1000    &   0.384    &   0.243    &   0.156    &   0.110    &   0.086    &   0.047    &   0.036    &   0.013    \\ 
10000    &   0.397    &   0.231    &   0.150    &   0.107    & {\bf  0.083}   & {\bf  0.043}   & {\bf  0.032}   & {\bf  0.012}
\end{tabular} 
}
\parbox{.45\linewidth}{\centering 
  discrete $ k =  50 $, $\alpha =  2 $,  $m =  10^5$, $rep =  500 $,   NRMSE 2-pass  \\
\begin{tabular}{||r | r r r r r r r r}
\hline 
 $\ell$, $T$    &  1  &  5  &  20  &  50  &  100  &  500  &  1000  &  10000  \\ 
\hline 
1    & {\bf  0.145}   &   0.172    &   0.235    &   0.290    &   0.345    &   0.505    &   0.601    &   1.063    \\ 
5    &   0.165    & {\bf  0.132}   &   0.145    &   0.169    &   0.201    &   0.311    &   0.370    &   0.636    \\ 
20    &   0.205    &   0.147    &   0.122    &   0.125    &   0.137    &   0.196    &   0.232    &   0.422    \\ 
50    &   0.241    &   0.165    & {\bf  0.120}   &   0.104    &   0.101    &   0.132    &   0.158    &   0.282    \\ 
100    &   0.294    &   0.194    &   0.140    &   0.112    &   0.097    &   0.094    &   0.105    &   0.184    \\ 
500    &   0.320    &   0.202    &   0.138    &   0.105    &   0.082    &   0.042    &   0.031    &   0.016    \\ 
1000    &   0.334    &   0.216    &   0.146    &   0.109    &   0.084    &   0.041    &   0.029    &   0.010    \\ 
10000    &   0.326    &   0.206    &   0.140    & {\bf  0.104}   & {\bf  0.081}   & {\bf  0.040}   & {\bf  0.028}   & {\bf  0.010}
\end{tabular} 
}
}
\caption{Simulation Results for Discrete SH$_\ell$ \label{DiscreteRes:fig}}
\end{figure*}

\begin{figure*}[htbp]
{\tiny
\notinproc{
\parbox{.45\linewidth}{\centering
  continuous $ k =  100 $, $\alpha =  1.1 $,  $m =  100000 $, $rep =  500 $,   relerr  1-pass  \\
\begin{tabular}{r | r r r r r r r r}
\hline 
 $\ell$, $T$    &  1  &  5  &  20  &  50  &  100  &  500  &  1000  &  10000  \\ 
\hline 
0.1    & {\bf  0.079}   &   0.094    &   0.140    &   0.196    &   0.248    &   0.394    &   0.456    &   0.665    \\ 
1    &   0.079    &   0.086    &   0.119    &   0.164    &   0.210    &   0.356    &   0.432    &   0.640    \\ 
5    &   0.086    & {\bf  0.079}   &   0.086    &   0.106    &   0.131    &   0.243    &   0.310    &   0.498    \\ 
20    &   0.089    &   0.084    &   0.084    &   0.089    &   0.098    &   0.153    &   0.196    &   0.394    \\ 
50    &   0.090    &   0.083    & {\bf  0.081}   & {\bf  0.081}   &   0.085    &   0.114    &   0.139    &   0.295    \\ 
100    &   0.099    &   0.089    &   0.083    &   0.082    & {\bf  0.081}   &   0.089    &   0.104    &   0.218    \\ 
500    &   0.111    &   0.102    &   0.094    &   0.093    &   0.090    &   0.083    &   0.082    &   0.096    \\ 
1000    &   0.114    &   0.104    &   0.097    &   0.092    &   0.087    &   0.080    &   0.076    &   0.065    \\ 
10000    &   0.106    &   0.097    &   0.091    &   0.086    &   0.084    & {\bf  0.076}   & {\bf  0.073}   & {\bf  0.062}   \\ 
\end{tabular} 
}
\parbox{.45\linewidth}{\centering
  continuous $ k =  100 $, $\alpha =  1.1 $,  $m =  100000 $, $rep =  500 $,   rel err  2-pass  \\
\begin{tabular}{||r | r r r r r r r r}
\hline 
 $\ell$, $T$    &  1  &  5  &  20  &  50  &  100  &  500  &  1000  &  10000  \\ 
\hline 
0.1    &   0.079    &   0.094    &   0.141    &   0.196    &   0.249    &   0.394    &   0.456    &   0.665    \\ 
1    & {\bf  0.078}   &   0.085    &   0.118    &   0.163    &   0.210    &   0.355    &   0.431    &   0.640    \\ 
5    &   0.083    & {\bf  0.079}   &   0.087    &   0.105    &   0.132    &   0.245    &   0.310    &   0.499    \\ 
20    &   0.087    &   0.084    &   0.083    &   0.088    &   0.097    &   0.153    &   0.196    &   0.393    \\ 
50    &   0.087    &   0.083    & {\bf  0.081}   & {\bf  0.080}   &   0.084    &   0.114    &   0.139    &   0.295    \\ 
100    &   0.096    &   0.088    &   0.083    &   0.081    & {\bf  0.080}   &   0.088    &   0.103    &   0.219    \\ 
500    &   0.107    &   0.099    &   0.094    &   0.091    &   0.089    &   0.083    &   0.081    &   0.097    \\ 
1000    &   0.109    &   0.101    &   0.094    &   0.090    &   0.086    &   0.079    &   0.076    &   0.065    \\ 
10000    &   0.103    &   0.094    &   0.087    &   0.085    &   0.083    & {\bf  0.076}   & {\bf  0.072}   & {\bf  0.062}   \\ 
\end{tabular} 
}}
}

\noindent
{\tiny
\parbox{.45\linewidth}{\centering
  continuous $ k =  100 $, $\alpha =  1.1 $,  $m =  100000 $, $rep =  500 $,   NRMSE 1-pass  \\
\begin{tabular}{r | r r r r r r r r}
\hline 
 $\ell$, $T$    &  1  &  5  &  20  &  50  &  100  &  500  &  1000  &  10000  \\ 
\hline 
0.1    &   0.098    &   0.118    &   0.180    &   0.252    &   0.326    &   0.648    &   0.897    &   2.399    \\ 
1    & {\bf  0.098}   &   0.108    &   0.150    &   0.207    &   0.267    &   0.541    &   0.781    &   2.006    \\ 
5    &   0.109    & {\bf  0.100}   &   0.110    &   0.135    &   0.170    &   0.316    &   0.432    &   1.135    \\ 
20    &   0.114    &   0.106    &   0.105    &   0.112    &   0.126    &   0.198    &   0.252    &   0.672    \\ 
50    &   0.117    &   0.106    & {\bf  0.103}   &   0.105    &   0.108    &   0.145    &   0.179    &   0.418    \\ 
100    &   0.125    &   0.112    &   0.103    & {\bf  0.101}   & {\bf  0.101}   &   0.114    &   0.133    &   0.285    \\ 
500    &   0.141    &   0.130    &   0.122    &   0.119    &   0.115    &   0.106    &   0.103    &   0.120    \\ 
1000    &   0.144    &   0.133    &   0.123    &   0.118    &   0.112    &   0.102    &   0.097    &   0.083    \\ 
10000    &   0.133    &   0.121    &   0.115    &   0.110    &   0.108    & {\bf  0.098}   & {\bf  0.094}   & {\bf  0.080}   \\ 
\end{tabular} 
}
\parbox{.45\linewidth}{\centering
  continuous $ k =  100 $, $\alpha =  1.1 $,  $m =  100000 $, $rep =  500 $,   NRMSE 2-pass  \\
\begin{tabular}{|| r | r r r r r r r r}
\hline 
 $\ell$, $T$    &  1  &  5  &  20  &  50  &  100  &  500  &  1000  &  10000  \\ 
\hline 
0.1    &   0.098    &   0.118    &   0.180    &   0.252    &   0.326    &   0.649    &   0.897    &   2.399    \\ 
1    & {\bf  0.097}   &   0.107    &   0.149    &   0.206    &   0.266    &   0.540    &   0.780    &   2.005    \\ 
5    &   0.106    & {\bf  0.100}   &   0.111    &   0.135    &   0.171    &   0.316    &   0.433    &   1.135    \\ 
20    &   0.112    &   0.106    &   0.104    &   0.111    &   0.125    &   0.197    &   0.251    &   0.671    \\ 
50    &   0.112    &   0.106    & {\bf  0.103}   &   0.103    &   0.108    &   0.144    &   0.178    &   0.418    \\ 
100    &   0.120    &   0.111    &   0.103    & {\bf  0.101}   & {\bf  0.100}   &   0.113    &   0.131    &   0.285    \\ 
500    &   0.137    &   0.128    &   0.121    &   0.117    &   0.114    &   0.105    &   0.103    &   0.120    \\ 
1000    &   0.138    &   0.130    &   0.121    &   0.115    &   0.111    &   0.101    &   0.097    &   0.082    \\ 
10000    &   0.130    &   0.119    &   0.112    &   0.109    &   0.106    & {\bf  0.097}   & {\bf  0.093}   & {\bf  0.079}   \\ 
\end{tabular} 
}}

\notinproc{
\noindent
{\tiny
\parbox{.45\linewidth}{\centering
  continuous $ k =  100 $, $\alpha =  1.2 $,  $m =  100000 $, $rep =  500 $,   relerr  1-pass  \\
\begin{tabular}{r | r r r r r r r r}
\hline 
 $\ell$, $T$    &  1  &  5  &  20  &  50  &  100  &  500  &  1000  &  10000  \\ 
\hline 
0.1    &   0.079    &   0.101    &   0.153    &   0.216    &   0.281    &   0.493    &   0.601    &   1.053    \\ 
1    & {\bf  0.079}   &   0.089    &   0.127    &   0.178    &   0.232    &   0.427    &   0.535    &   0.900    \\ 
5    &   0.087    & {\bf  0.079}   &   0.088    &   0.113    &   0.143    &   0.269    &   0.353    &   0.635    \\ 
20    &   0.097    &   0.083    & {\bf  0.074}   & {\bf  0.081}   &   0.095    &   0.163    &   0.208    &   0.477    \\ 
50    &   0.121    &   0.100    &   0.091    &   0.087    &   0.088    &   0.116    &   0.148    &   0.333    \\ 
100    &   0.121    &   0.102    &   0.091    &   0.084    & {\bf  0.081}   &   0.095    &   0.110    &   0.237    \\ 
500    &   0.129    &   0.110    &   0.099    &   0.093    &   0.089    &   0.074    &   0.068    &   0.063    \\ 
1000    &   0.138    &   0.115    &   0.099    &   0.094    &   0.091    &   0.073    &   0.066    &   0.049    \\ 
10000    &   0.135    &   0.117    &   0.101    &   0.090    &   0.085    & {\bf  0.070}   & {\bf  0.064}   & {\bf  0.048}   \\ 
\end{tabular} 
}
\parbox{.45\linewidth}{\centering
  continuous $ k =  100 $, $\alpha =  1.2 $,  $m =  100000 $, $rep =  500 $,   rel err  2-pass  \\
\begin{tabular}{||r | r r r r r r r r}
\hline 
 $\ell$, $T$    &  1  &  5  &  20  &  50  &  100  &  500  &  1000  &  10000  \\ 
\hline 
0.1    &   0.078    &   0.100    &   0.153    &   0.217    &   0.281    &   0.493    &   0.601    &   1.053    \\ 
1    & {\bf  0.077}   &   0.088    &   0.128    &   0.179    &   0.233    &   0.426    &   0.535    &   0.900    \\ 
5    &   0.084    & {\bf  0.078}   &   0.088    &   0.113    &   0.142    &   0.269    &   0.353    &   0.635    \\ 
20    &   0.094    &   0.081    & {\bf  0.074}   & {\bf  0.081}   &   0.095    &   0.163    &   0.209    &   0.478    \\ 
50    &   0.116    &   0.098    &   0.090    &   0.086    &   0.087    &   0.114    &   0.146    &   0.333    \\ 
100    &   0.117    &   0.101    &   0.089    &   0.083    & {\bf  0.081}   &   0.092    &   0.108    &   0.237    \\ 
500    &   0.121    &   0.106    &   0.096    &   0.090    &   0.086    &   0.074    &   0.068    &   0.063    \\ 
1000    &   0.129    &   0.110    &   0.098    &   0.092    &   0.087    &   0.073    &   0.066    &   0.048    \\ 
10000    &   0.128    &   0.112    &   0.097    &   0.089    &   0.083    & {\bf  0.069}   & {\bf  0.063}   & {\bf  0.046}   \\ 
\end{tabular} 
}}

\noindent
{\tiny
\parbox{.45\linewidth}{\centering
  continuous $ k =  100 $, $\alpha =  1.2 $,  $m =  100000 $, $rep =  500 $,   NRMSE 1-pass  \\
\begin{tabular}{r | r r r r r r r r}
\hline 
 $\ell$, $T$    &  1  &  5  &  20  &  50  &  100  &  500  &  1000  &  10000  \\ 
\hline 
0.1    & {\bf  0.099}   &   0.126    &   0.189    &   0.267    &   0.348    &   0.676    &   0.942    &   3.061    \\ 
1    &   0.099    &   0.111    &   0.161    &   0.225    &   0.291    &   0.565    &   0.788    &   2.190    \\ 
5    &   0.109    & {\bf  0.100}   &   0.111    &   0.142    &   0.182    &   0.336    &   0.444    &   1.070    \\ 
20    &   0.122    &   0.105    & {\bf  0.094}   & {\bf  0.102}   &   0.118    &   0.202    &   0.261    &   0.649    \\ 
50    &   0.150    &   0.127    &   0.115    &   0.111    &   0.111    &   0.148    &   0.187    &   0.416    \\ 
100    &   0.153    &   0.128    &   0.116    &   0.106    & {\bf  0.102}   &   0.117    &   0.137    &   0.293    \\ 
500    &   0.161    &   0.139    &   0.125    &   0.116    &   0.111    &   0.092    &   0.086    &   0.081    \\ 
1000    &   0.173    &   0.145    &   0.125    &   0.117    &   0.111    &   0.093    &   0.085    &   0.062    \\ 
10000    &   0.169    &   0.142    &   0.125    &   0.113    &   0.106    & {\bf  0.087}   & {\bf  0.079}   & {\bf  0.059}   \\ 
\end{tabular} 
}
\parbox{.45\linewidth}{\centering
  continuous $ k =  100 $, $\alpha =  1.2 $,  $m =  100000 $, $rep =  500 $,   NRMSE 2-pass  \\
\begin{tabular}{||r | r r r r r r r r}
\hline 
 $\ell$, $T$    &  1  &  5  &  20  &  50  &  100  &  500  &  1000  &  10000  \\ 
\hline 
0.1    & {\bf  0.097}   &   0.125    &   0.189    &   0.267    &   0.348    &   0.677    &   0.942    &   3.061    \\ 
1    &   0.097    &   0.111    &   0.162    &   0.225    &   0.291    &   0.565    &   0.787    &   2.190    \\ 
5    &   0.106    & {\bf  0.099}   &   0.110    &   0.141    &   0.181    &   0.336    &   0.443    &   1.070    \\ 
20    &   0.119    &   0.102    & {\bf  0.093}   & {\bf  0.101}   &   0.119    &   0.203    &   0.262    &   0.650    \\ 
50    &   0.145    &   0.125    &   0.114    &   0.110    &   0.111    &   0.146    &   0.185    &   0.415    \\ 
100    &   0.147    &   0.127    &   0.113    &   0.105    & {\bf  0.102}   &   0.114    &   0.135    &   0.293    \\ 
500    &   0.152    &   0.134    &   0.122    &   0.115    &   0.109    &   0.092    &   0.085    &   0.080    \\ 
1000    &   0.161    &   0.138    &   0.123    &   0.114    &   0.108    &   0.092    &   0.084    &   0.060    \\ 
10000    &   0.159    &   0.138    &   0.120    &   0.110    &   0.103    & {\bf  0.087}   & {\bf  0.079}   & {\bf  0.057}   \\ 
\end{tabular} 
}}

\noindent
{\tiny
\parbox{.45\linewidth}{\centering
  continuous $ k =  100 $, $\alpha =  1.5 $,  $m =  100000 $, $rep =  500 $,   relerr  1-pass  \\
\begin{tabular}{r | r r r r r r r r}
\hline 
 $\ell$, $T$    &  1  &  5  &  20  &  50  &  100  &  500  &  1000  &  10000  \\ 
\hline 
0.1    & {\bf  0.083}   &   0.108    &   0.156    &   0.208    &   0.267    &   0.444    &   0.555    &   0.990    \\ 
1    &   0.083    &   0.096    &   0.132    &   0.177    &   0.230    &   0.379    &   0.456    &   0.900    \\ 
5    &   0.094    & {\bf  0.078}   &   0.091    &   0.113    &   0.137    &   0.236    &   0.304    &   0.622    \\ 
20    &   0.122    &   0.090    & {\bf  0.077}   &   0.080    &   0.090    &   0.142    &   0.176    &   0.368    \\ 
50    &   0.151    &   0.107    &   0.081    & {\bf  0.072}   &   0.073    &   0.097    &   0.118    &   0.235    \\ 
100    &   0.172    &   0.118    &   0.090    &   0.072    & {\bf  0.065}   &   0.062    &   0.071    &   0.137    \\ 
500    &   0.178    &   0.131    &   0.101    &   0.082    &   0.071    & {\bf  0.046}   &   0.038    &   0.020    \\ 
1000    &   0.180    &   0.131    &   0.097    &   0.083    &   0.070    &   0.047    & {\bf  0.038}   & {\bf  0.019}   \\ 
10000    &   0.182    &   0.128    &   0.103    &   0.085    &   0.072    &   0.047    &   0.038    &   0.020    \\ 
\end{tabular} 
}
\parbox{.45\linewidth}{\centering
  continuous $ k =  100 $, $\alpha =  1.5 $,  $m =  100000 $, $rep =  500 $,   rel err  2-pass  \\
\begin{tabular}{||r | r r r r r r r r}
\hline 
 $\ell$, $T$    &  1  &  5  &  20  &  50  &  100  &  500  &  1000  &  10000  \\ 
\hline 
0.1    &   0.083    &   0.108    &   0.157    &   0.208    &   0.266    &   0.444    &   0.555    &   0.990    \\ 
1    & {\bf  0.081}   &   0.095    &   0.132    &   0.177    &   0.229    &   0.379    &   0.455    &   0.900    \\ 
5    &   0.090    & {\bf  0.077}   &   0.089    &   0.112    &   0.136    &   0.235    &   0.304    &   0.622    \\ 
20    &   0.114    &   0.086    & {\bf  0.077}   &   0.079    &   0.089    &   0.141    &   0.175    &   0.368    \\ 
50    &   0.131    &   0.100    &   0.079    & {\bf  0.072}   &   0.072    &   0.095    &   0.117    &   0.236    \\ 
100    &   0.151    &   0.109    &   0.086    &   0.072    & {\bf  0.064}   &   0.060    &   0.069    &   0.137    \\ 
500    &   0.161    &   0.122    &   0.095    &   0.080    &   0.069    & {\bf  0.045}   & {\bf  0.036}   & {\bf  0.018}   \\ 
1000    &   0.160    &   0.118    &   0.092    &   0.079    &   0.069    &   0.045    &   0.036    &   0.018    \\ 
10000    &   0.161    &   0.120    &   0.097    &   0.082    &   0.070    &   0.045    &   0.036    &   0.018    \\ 
\end{tabular} 
}}
} 

\noindent
{\tiny
\parbox{.45\linewidth}{\centering
  continuous $ k =  100 $, $\alpha =  1.5 $,  $m =  100000 $, $rep =  500 $,   NRMSE 1-pass  \\
\begin{tabular}{r | r r r r r r r r}
\hline 
 $\ell$, $T$    &  1  &  5  &  20  &  50  &  100  &  500  &  1000  &  10000  \\ 
\hline 
0.1    &   0.106    &   0.134    &   0.198    &   0.266    &   0.341    &   0.558    &   0.688    &   1.508    \\ 
1    & {\bf  0.103}   &   0.120    &   0.168    &   0.228    &   0.292    &   0.478    &   0.586    &   1.320    \\ 
5    &   0.119    & {\bf  0.096}   &   0.113    &   0.142    &   0.174    &   0.301    &   0.382    &   0.766    \\ 
20    &   0.152    &   0.115    & {\bf  0.096}   &   0.100    &   0.112    &   0.176    &   0.220    &   0.455    \\ 
50    &   0.190    &   0.136    &   0.102    & {\bf  0.092}   &   0.092    &   0.121    &   0.148    &   0.294    \\ 
100    &   0.214    &   0.152    &   0.115    &   0.092    & {\bf  0.082}   &   0.078    &   0.088    &   0.167    \\ 
500    &   0.225    &   0.169    &   0.129    &   0.105    &   0.089    &   0.059    &   0.049    &   0.025    \\ 
1000    &   0.224    &   0.163    &   0.122    &   0.102    &   0.088    & {\bf  0.059}   & {\bf  0.048}   & {\bf  0.024}   \\ 
10000    &   0.230    &   0.162    &   0.130    &   0.108    &   0.091    &   0.059    &   0.049    &   0.025    \\ 
\end{tabular} 
}
\parbox{.45\linewidth}{\centering
  continuous $ k =  100 $, $\alpha =  1.5 $,  $m =  100000 $, $rep =  500 $,   NRMSE 2-pass  \\
\begin{tabular}{||r | r r r r r r r r}
\hline 
 $\ell$, $T$    &  1  &  5  &  20  &  50  &  100  &  500  &  1000  &  10000  \\ 
\hline 
0.1    &   0.106    &   0.134    &   0.198    &   0.266    &   0.341    &   0.558    &   0.688    &   1.508    \\ 
1    & {\bf  0.101}   &   0.118    &   0.167    &   0.228    &   0.292    &   0.478    &   0.586    &   1.320    \\ 
5    &   0.113    & {\bf  0.096}   &   0.111    &   0.140    &   0.172    &   0.300    &   0.381    &   0.766    \\ 
20    &   0.142    &   0.110    & {\bf  0.095}   &   0.098    &   0.110    &   0.175    &   0.219    &   0.454    \\ 
50    &   0.168    &   0.126    &   0.100    & {\bf  0.091}   &   0.091    &   0.120    &   0.147    &   0.294    \\ 
100    &   0.191    &   0.139    &   0.109    &   0.091    & {\bf  0.082}   &   0.076    &   0.087    &   0.167    \\ 
500    &   0.203    &   0.154    &   0.121    &   0.102    &   0.088    &   0.057    &   0.045    &   0.022    \\ 
1000    &   0.198    &   0.146    &   0.117    &   0.099    &   0.086    & {\bf  0.056}   & {\bf  0.045}   & {\bf  0.022}   \\ 
10000    &   0.205    &   0.152    &   0.122    &   0.104    &   0.089    &   0.057    &   0.046    &   0.023    \\ 
\end{tabular} 
}}

\notinproc{
\noindent
{\tiny
\parbox{.45\linewidth}{\centering
  continuous $ k =  50 $, $\alpha =  1.8 $,  $m =  100000 $, $rep =  500 $,   relerr  1-pass  \\
\begin{tabular}{r | r r r r r r r r}
\hline 
 $\ell$, $T$    &  1  &  5  &  20  &  50  &  100  &  500  &  1000  &  10000  \\ 
\hline 
0.1    & {\bf  0.106}   &   0.139    &   0.197    &   0.258    &   0.310    &   0.463    &   0.561    &   1.009    \\ 
1    &   0.112    &   0.125    &   0.172    &   0.220    &   0.266    &   0.411    &   0.497    &   0.899    \\ 
5    &   0.135    & {\bf  0.104}   &   0.122    &   0.151    &   0.181    &   0.287    &   0.348    &   0.642    \\ 
20    &   0.205    &   0.133    & {\bf  0.108}   &   0.106    &   0.121    &   0.178    &   0.213    &   0.400    \\ 
50    &   0.258    &   0.167    &   0.120    & {\bf  0.103}   &   0.097    &   0.119    &   0.143    &   0.259    \\ 
100    &   0.295    &   0.196    &   0.133    &   0.107    & {\bf  0.092}   &   0.078    &   0.084    &   0.152    \\ 
500    &   0.324    &   0.214    &   0.153    &   0.118    &   0.094    &   0.054    &   0.041    &   0.017    \\ 
1000    &   0.335    &   0.220    &   0.152    &   0.118    &   0.096    &   0.055    &   0.042    &   0.018    \\ 
10000    &   0.339    &   0.200    &   0.141    &   0.118    &   0.093    & {\bf  0.053}   & {\bf  0.041}   & {\bf  0.017}   \\ 
\end{tabular} 
}
\parbox{.45\linewidth}{\centering
  continuous $ k =  50 $, $\alpha =  1.8 $,  $m =  100000 $, $rep =  500 $,   rel err  2-pass  \\
\begin{tabular}{||r | r r r r r r r r}
\hline 
 $\ell$, $T$    &  1  &  5  &  20  &  50  &  100  &  500  &  1000  &  10000  \\ 
\hline 
0.1    & {\bf  0.106}   &   0.139    &   0.198    &   0.258    &   0.310    &   0.464    &   0.561    &   1.009    \\ 
1    &   0.112    &   0.124    &   0.172    &   0.220    &   0.266    &   0.411    &   0.497    &   0.899    \\ 
5    &   0.129    & {\bf  0.102}   &   0.120    &   0.149    &   0.180    &   0.286    &   0.348    &   0.642    \\ 
20    &   0.174    &   0.127    & {\bf  0.106}   &   0.106    &   0.118    &   0.177    &   0.213    &   0.400    \\ 
50    &   0.215    &   0.151    &   0.118    & {\bf  0.103}   &   0.096    &   0.119    &   0.143    &   0.258    \\ 
100    &   0.250    &   0.171    &   0.125    &   0.104    & {\bf  0.090}   &   0.078    &   0.083    &   0.152    \\ 
500    &   0.279    &   0.194    &   0.140    &   0.112    &   0.092    &   0.052    &   0.039    &   0.015    \\ 
1000    &   0.280    &   0.191    &   0.143    &   0.115    &   0.095    &   0.053    &   0.039    &   0.015    \\ 
10000    &   0.279    &   0.187    &   0.136    &   0.110    &   0.091    & {\bf  0.051}   & {\bf  0.038}   & {\bf  0.015}   \\ 
\end{tabular} 
}}

\noindent
{\tiny
\parbox{.45\linewidth}{\centering
  continuous $ k =  50 $, $\alpha =  1.8 $,  $m =  100000 $, $rep =  500 $,   NRMSE 1-pass  \\
\begin{tabular}{r | r r r r r r r r}
\hline 
 $\ell$, $T$    &  1  &  5  &  20  &  50  &  100  &  500  &  1000  &  10000  \\ 
\hline 
0.1    & {\bf  0.135}   &   0.179    &   0.254    &   0.328    &   0.393    &   0.588    &   0.713    &   1.403    \\ 
1    &   0.142    &   0.161    &   0.220    &   0.284    &   0.342    &   0.518    &   0.623    &   1.166    \\ 
5    &   0.172    & {\bf  0.132}   &   0.151    &   0.189    &   0.227    &   0.360    &   0.434    &   0.782    \\ 
20    &   0.256    &   0.165    & {\bf  0.135}   &   0.133    &   0.152    &   0.227    &   0.275    &   0.498    \\ 
50    &   0.320    &   0.212    &   0.151    & {\bf  0.129}   &   0.123    &   0.153    &   0.181    &   0.325    \\ 
100    &   0.368    &   0.243    &   0.166    &   0.133    & {\bf  0.115}   &   0.098    &   0.106    &   0.190    \\ 
500    &   0.413    &   0.275    &   0.193    &   0.149    &   0.118    &   0.068    &   0.053    &   0.022    \\ 
1000    &   0.418    &   0.281    &   0.191    &   0.147    &   0.122    &   0.070    &   0.054    &   0.023    \\ 
10000    &   0.423    &   0.259    &   0.184    &   0.149    &   0.118    & {\bf  0.066}   & {\bf  0.052}   & {\bf  0.021}   \\ 
\end{tabular} 
}
\parbox{.45\linewidth}{\centering
  continuous $ k =  50 $, $\alpha =  1.8 $,  $m =  100000 $, $rep =  500 $,   NRMSE 2-pass  \\
\begin{tabular}{||r | r r r r r r r r}
\hline 
 $\ell$, $T$    &  1  &  5  &  20  &  50  &  100  &  500  &  1000  &  10000  \\ 
\hline 
0.1    & {\bf  0.134}   &   0.179    &   0.254    &   0.328    &   0.394    &   0.588    &   0.713    &   1.403    \\ 
1    &   0.142    &   0.160    &   0.220    &   0.284    &   0.342    &   0.518    &   0.623    &   1.166    \\ 
5    &   0.164    & {\bf  0.131}   &   0.148    &   0.186    &   0.226    &   0.360    &   0.433    &   0.782    \\ 
20    &   0.219    &   0.156    & {\bf  0.131}   &   0.134    &   0.149    &   0.227    &   0.275    &   0.498    \\ 
50    &   0.269    &   0.191    &   0.149    & {\bf  0.129}   &   0.122    &   0.152    &   0.180    &   0.325    \\ 
100    &   0.317    &   0.214    &   0.157    &   0.129    & {\bf  0.112}   &   0.097    &   0.105    &   0.191    \\ 
500    &   0.353    &   0.245    &   0.176    &   0.142    &   0.117    &   0.066    &   0.049    &   0.019    \\ 
1000    &   0.353    &   0.242    &   0.180    &   0.144    &   0.119    &   0.067    &   0.050    &   0.019    \\ 
10000    &   0.358    &   0.239    &   0.173    &   0.139    &   0.115    & {\bf  0.064}   & {\bf  0.047}   & {\bf  0.018}   \\ 
\end{tabular} 
}}

\noindent
{\tiny
\parbox{.45\linewidth}{\centering
  continuous $ k =  50 $, $\alpha =  2 $,  $m =  100000 $, $rep =  500 $,   relerr  1-pass  \\
\begin{tabular}{r | r r r r r r r r}
\hline 
 $\ell$, $T$    &  1  &  5  &  20  &  50  &  100  &  500  &  1000  &  10000  \\ 
\hline 
0.1    & {\bf  0.100}   &   0.127    &   0.172    &   0.214    &   0.254    &   0.403    &   0.481    &   0.851    \\ 
1    &   0.103    &   0.112    &   0.152    &   0.195    &   0.234    &   0.360    &   0.425    &   0.750    \\ 
5    &   0.149    & {\bf  0.104}   &   0.114    &   0.138    &   0.161    &   0.239    &   0.282    &   0.508    \\ 
20    &   0.217    &   0.135    & {\bf  0.099}   &   0.094    &   0.099    &   0.142    &   0.168    &   0.303    \\ 
50    &   0.271    &   0.163    &   0.110    & {\bf  0.084}   &   0.074    &   0.074    &   0.084    &   0.145    \\ 
100    &   0.307    &   0.185    &   0.111    &   0.084    &   0.066    &   0.035    &   0.026    &   0.011    \\ 
500    &   0.314    &   0.190    &   0.125    &   0.090    &   0.071    &   0.036    &   0.026    &   0.010    \\ 
1000    &   0.324    &   0.183    &   0.119    &   0.084    & {\bf  0.066}   & {\bf  0.033}   & {\bf  0.024}   & {\bf  0.009}   \\ 
10000    &   0.316    &   0.194    &   0.121    &   0.090    &   0.068    &   0.035    &   0.025    &   0.009    \\ 
\end{tabular} 
}
\parbox{.45\linewidth}{\centering
  continuous $ k =  50 $, $\alpha =  2 $,  $m =  100000 $, $rep =  500 $,   rel err  2-pass  \\
\begin{tabular}{|| r | r r r r r r r r}
\hline 
 $\ell$, $T$    &  1  &  5  &  20  &  50  &  100  &  500  &  1000  &  10000  \\ 
\hline 
0.1    & {\bf  0.099}   &   0.127    &   0.172    &   0.214    &   0.254    &   0.403    &   0.481    &   0.851    \\ 
1    &   0.100    &   0.111    &   0.151    &   0.194    &   0.234    &   0.360    &   0.425    &   0.750    \\ 
5    &   0.137    & {\bf  0.104}   &   0.113    &   0.137    &   0.160    &   0.239    &   0.282    &   0.508    \\ 
20    &   0.186    &   0.128    & {\bf  0.098}   &   0.092    &   0.098    &   0.142    &   0.167    &   0.303    \\ 
50    &   0.223    &   0.142    &   0.104    &   0.083    &   0.073    &   0.073    &   0.083    &   0.145    \\ 
100    &   0.257    &   0.161    &   0.110    & {\bf  0.082}   &   0.065    &   0.031    &   0.022    &   0.009    \\ 
500    &   0.254    &   0.170    &   0.120    &   0.090    &   0.069    &   0.032    &   0.023    &   0.008    \\ 
1000    &   0.255    &   0.167    &   0.113    &   0.083    & {\bf  0.063}   & {\bf  0.030}   & {\bf  0.022}   & {\bf  0.007}   \\ 
10000    &   0.256    &   0.166    &   0.114    &   0.086    &   0.067    &   0.032    &   0.023    &   0.008    \\ 
\end{tabular} 
}}
} 

\noindent
{\tiny
\parbox{.45\linewidth}{\centering
  continuous $ k =  50 $, $\alpha =  2 $,  $m =  100000 $, $rep =  500 $,   NRMSE 1-pass  \\
\begin{tabular}{r | r r r r r r r r}
\hline 
 $\ell$, $T$    &  1  &  5  &  20  &  50  &  100  &  500  &  1000  &  10000  \\ 
\hline 
0.1    & {\bf  0.126}   &   0.159    &   0.216    &   0.274    &   0.326    &   0.502    &   0.597    &   1.061    \\ 
1    &   0.129    &   0.141    &   0.192    &   0.244    &   0.293    &   0.449    &   0.526    &   0.908    \\ 
5    &   0.193    & {\bf  0.138}   &   0.146    &   0.173    &   0.202    &   0.300    &   0.353    &   0.626    \\ 
20    &   0.277    &   0.169    & {\bf  0.124}   &   0.118    &   0.125    &   0.183    &   0.216    &   0.377    \\ 
50    &   0.339    &   0.206    &   0.140    &   0.108    &   0.094    &   0.096    &   0.108    &   0.182    \\ 
100    &   0.390    &   0.236    &   0.146    & {\bf  0.107}   &   0.085    &   0.046    &   0.034    &   0.022    \\ 
500    &   0.397    &   0.250    &   0.162    &   0.114    &   0.092    &   0.047    &   0.034    &   0.012    \\ 
1000    &   0.396    &   0.232    &   0.150    &   0.108    & {\bf  0.083}   & {\bf  0.042}   & {\bf  0.031}   & {\bf  0.011}   \\ 
10000    &   0.404    &   0.244    &   0.155    &   0.114    &   0.085    &   0.043    &   0.032    &   0.012    \\ 
\end{tabular} 
}
\parbox{.45\linewidth}{\centering
  continuous $ k =  50 $, $\alpha =  2 $,  $m =  100000 $, $rep =  500 $,   NRMSE 2-pass  \\
\begin{tabular}{|| r | r r r r r r r r}
\hline 
 $\ell$, $T$    &  1  &  5  &  20  &  50  &  100  &  500  &  1000  &  10000  \\ 
\hline 
0.1    & {\bf  0.125}   &   0.159    &   0.216    &   0.274    &   0.326    &   0.502    &   0.597    &   1.061    \\ 
1    &   0.127    &   0.139    &   0.190    &   0.244    &   0.293    &   0.449    &   0.526    &   0.908    \\ 
5    &   0.178    & {\bf  0.137}   &   0.144    &   0.172    &   0.202    &   0.300    &   0.353    &   0.626    \\ 
20    &   0.235    &   0.163    & {\bf  0.123}   &   0.116    &   0.125    &   0.183    &   0.216    &   0.378    \\ 
50    &   0.282    &   0.184    &   0.133    &   0.106    &   0.093    &   0.094    &   0.106    &   0.181    \\ 
100    &   0.327    &   0.204    &   0.140    &   0.105    &   0.083    &   0.041    &   0.030    &   0.020    \\ 
500    &   0.321    &   0.218    &   0.152    &   0.114    &   0.089    &   0.042    &   0.030    &   0.010    \\ 
1000    &   0.322    &   0.208    &   0.143    & {\bf  0.105}   & {\bf  0.080}   & {\bf  0.039}   & {\bf  0.028}   & {\bf  0.009}   \\ 
10000    &   0.326    &   0.213    &   0.147    &   0.109    &   0.084    &   0.040    &   0.028    &   0.010    \\ 
\end{tabular} 
}}
\caption{Simulation Results for Continuous SH$_\ell$ \label{ContRes:fig}}
\end{figure*}

\section{Simulations} \label{experiments:sec}

 Our experimental evaluation is aimed to understand the error 
  distribution of our estimators.  Our analysis provided statistical 
  guarantees on  the errors that are close to the ``gold standard'' attainable 
  on aggregated data.  The analysis, however is worst-case in terms of 
  the dependence on the disparity $\max\{\ell/T,T/\ell\}$, the 
factors of $(e/(e-1))^{0.5}   \approx 1.26$ (2-pass) and $(2e/(e-1))^{0.5}
\approx 1.8$ (1-pass),  which assume a worst-case frequency distribution 
(error is larger when $w_x   \approx \ell$), and not reflecting  the advantage of with-replacement 
sampling that is significant when there is skew.   We therefore 
expect actual errors to be much lower than our upper bounds. 

Our sampling algorithms and estimators were implemented in Python 
using numpy.random and hashlib libraries.  
Simulations were performed on MacBook Air and Mac mini computers.
We did not attempt  to benchmark performance in terms of running time, since 
  computationally, our algorithms are similar to  the 
widely applied distinct sampling or counting algorithms and can easily be tuned and scaled to very 
large data sets and common platforms.

We generated streams of $10^5$ elements with uniform weights.  The 
keys were drawn from a Zipf distribution with parameter $\alpha = 1,1. 1.2, 1.5, 1.8,
  2$.  This range of Zipf parameters is typical to large data sets and 
  working with them allowed us to finely understand the 
  error dependence on the skew (Zipf with larger $\alpha$ is more skewed and 
  has fewer distinct keys per number of elements).  
 The average number of distinct keys in our 
  simulations, and respective sample sizes we used,
 was $4.3 \times 10^4$ for $\alpha=1.1$ (used $k=100$);  $1.84 \times 10^4$ for $\alpha=1.2$ (used $k=100$);
  $ 3.04 \times 10^3$ for $\alpha=1.5$ (used $k=100$);
$841$ for $\alpha=1.8$ (used $k=50$); and 
$437$ for $\alpha=2$ (used $k=50$). 

  For each stream, we computed the exact frequencies of each key for 
  reference in the  error computation of the estimates. 
For a set of sample cap parameters $\ell = 1,5, 20,
  50, 100, 1000, 10000$ (and also $\ell =
  0.1$ with continuous samples), 
we  computed discrete and continuous fixed-$k$ SH$_\ell$ samples. 
 Discrete SH$_\ell$ sampling  used Algorithm \ref{ellaSH:alg}  with scoring 
 function \eqref{delementscore:eq} and continuous SH$_\ell$ sampling used 
 Algorithm \ref{cellaSH:alg}. 

   From each sample, we computed an estimate of the frequency cap 
   statistic $Q(\Cap_T,{\cal X})$ over all keys,  for parameters $T =$ 1, 5, 20, 50, 100, 1000,
   10000.  With discrete SH$_\ell$, we used the estimator of the 
   form \eqref{coefform:eq}  and computed estimation coefficients as 
   in Theorem~\ref{discreteestcoef:thm}. 
With  continuous SH$_\ell$,  we used the 
   estimator \eqref{contest:eq} with coefficient function \eqref{contbeta:eq}, which for 
  $\Cap_T$ statistics is:
$$\beta(c) = \frac{\min\{T,c\}}{\min\{1,\ell\tau\}} + \tau^{-1}
I_{c<T} \ .$$

For each $\ell,T$ combination, we also computed the estimate that is obtained from 2-pass 
   algorithms (Section~\ref{2pass:sec}), 
applied with element scoring \eqref{delementscore:eq} for discrete schemes 
   and \eqref{elementscorecont:eq} for continuous schemes.  We used 
   the inverse probability estimate $\sum_x \min\{T,w_x\}/\Phi(w_x)$,
   where $\Phi(w_x)$ is \eqref{contcumphi} for continuous schemes and 
   as outlined in Section \ref{discrete:sec} for discrete schemes. 

  For each of these estimates, we computed the 
   relative and NRMSE errors, averaged over multiple 
   ($rep=200$ or $rep=500$)  simulations (each using a fresh hash 
   function and randomness).   
Selected simulation results showing the errors for $\ell,T$
combinations are 
provided in 
Figure~\ref{DiscreteRes:fig} for discrete SH$_\ell$  and 
in Figure~\ref{ContRes:fig} for continuous SH$_\ell$.  The minimum 
error for each statistics $T$ across samples $\ell$ is boldfaced. 
\onlyinproc{Additional results are provided in the TR \cite{freqCap:2015}.}

\notinproc{\subsection*{Discussion of results}}\onlyinproc{\noindent
  {\bf Discussion:}}
 When looking at the parameter $\ell$ with smallest error for each cap
 statistics $T$, we see the diagonal pattern expected from 
 our analysis, where the error is minimized when $\ell\approx T$ and
 degrades with disparity between $T$ and $\ell$.
Note that the smallest distinct sampling threshold 
we had was $\tau \approx 0.001$ (for $\alpha=1.1$), therefore, our high 
  $\ell $ values effectively emulated uncapped SH.

Even for  these realistic distributions, we observe that 
a considerable performance gain by 
  using an appropriate  sample for our particular cap statistics. 
We can also see that the sensitivity of the error to the parameter $\ell$ increases with
skew (higher Zipf parameter $\alpha$).
In particular, the ratio of the error to the boldfaced minimum when
 using a high $\ell$ sample to estimate distinct counts
was up to a factor of  3 whereas the reverse could be 30 fold or more.
The increase in error for mid-cap statistics by using the better one of
$\ell=1,\infty$ instead of the minimum was up to 40\%.  Note however
that even this is optimistic, as we measured error on the whole
population -- on segments with frequency  distributions that do not match that of
the population, error can be much higher.\footnote{
 To make this point clearer, our selected segment was
  the whole population, which  means that for the segment, the number of samples was the
  same as when sampling using $\ell = T$.  Estimation quality
  deterioration from disparity was only due to the allocation of
  sampling probabilities within segment.  We can expect worst results
  (but again, theory bounds the worst-case pretty tightly), when
  adversely selecting segments.  For example, 
sampling a skewed distribution with very large $\ell$  and choosing
segments with small $w_x=1$ and $T=1$.  In this case, the segment can
have a high fraction of distinct keys but a small fraction of total
weight and will obtain very few samples. }


  Comparing the error of 2-pass versus streaming estimates (both are the
  same for distinct counts $\ell=1$ but diverge otherwise), we
  observe that the benefit of the second pass is limited to 10\% and
typically lower.  This agrees with our CV upper bounds which are only
slightly larger for the streaming estimates.
This suggests that the choice of scheme should depend on the computational platform.

  The $\ell=1, T=1$ estimates have NRMSE $\approx 1/\sqrt{k-2}$, this
  is because the upper bounds for
  approximate distinct counting are fairly tight~\cite{ECohen6f,ECohenADS:PODS2014}  as there is no
  dependence on the frequency distribution.    In our simulations, for
  higher cap values $T$,   the  minimum error (over $\ell$) 
was typically much lower than the CV upper bounds.
  This suggests using adaptive confidence bounds, based on
  sampled frequencies, rather than relying only on the CV upper bounds.

\section{Related Work}
 There is a large body of work on computing statistics over unaggregated data  which
 we can not hope to cover here.  The toolbox includes
deterministic 
  algorithms \cite{MisraGries:1982}, other sampling algorithms
  \cite{incsum:VLDB2009}, and 
Linear sketches (random linear projections)
\notinproc{\cite{JLlemma:1984,ams99,indyk:stable,CormodeMuthu:2005}}
\onlyinproc{\cite{JLlemma:1984,ams99,indyk:stable}.}
\notinproc{
Deterministic algorithms work well for approximate heavy hitters and quantiles.
Linear sketches project the
key-weight vectors to a lower dimensional vector. 
Linearity implies efficient updates of the sketch when processing
elements in a streaming setting.}
Most related to frequency cap statistics are
sketches based on $p$-stable distributions that are designed to estimate frequency moments for $p\in [0,1]$ \cite{indyk:stable}
and $L_p$ sampling \cite{MoWo:SODA2010,Jowhari:Saglam:Tardos:11}. 
These techniques do not apply for cap statistics, as there 
are no appropriate stable distributions for cap functions.  
\notinproc{
They are also
specific to the choice of $p$ and there is no 
support for segment queries.  $L_p$ samples,
which sample keys roughly proportionally to $w_x^p$, are
with-replacement, so less effective for skewed data,  and have
polylogarithmic encoding overhead.  
Of relevance to
us is also a characterization of all monotone frequency statistics  that can be estimated in
polylogarithmic space and a single pass \cite{BravermanOstro:STOC2010}.  The construction, however, is mostly
of theoretical interest.
Generally, linear sketches have a
significant encoding overhead and 
in practice,  when updates are positive, are outperformed by
sample-based sketches.
In particular, all practical distinct counting algorithms are based on the 
sample-based MinHash sketches 
\cite{FlajoletMartin85,hyperloglog:2007,hyperloglogpractice:EDBT2013,ECohenADS:PODS2014}
and for sum queries,  weighted SH experimentally
dominated linear sketches even in the presence of some negative 
updates \cite{CCD:sigmetrics12}.  
}

\subsection*{Conclusion}


  Frequency cap statistics are fundamental to data analysis.
  We propose a principled and practical sampling solution for scalably
  and accurately estimating frequency cap
  statistics over unaggregated data sets.  The sample is computed
  using state proportional to the specified desired sample size and the estimates have
  error bounds that nearly match those that can be obtained by
  an optimal weighted sample of the same size that can only be computed over the aggregated view.
Our design brings the benefits of approximate distinct counters,
 which are extensively deployed in the industry, to general
 frequency cap statistics.

 Looking ahead, we would like to apply
our framework for sampling unaggregated
 data sets to other statistics, extend it to support negative updates
 \cite{GemullaLH:vldb06,CCD:sigmetrics12}, and understand the
 theoretical boundaries of the approach.


\notinproc{
\section{Acknowledgement}

}\onlyinproc{\smallskip \noindent {\bf Acknowledgement:}} The author would like to thank Kevin Lang for bringing to her
  attention the use of  frequency capping in online advertising and the need for
 efficient sketches that support it.

{\onlyinproc{\small}\notinproc{\small}
\bibliographystyle{plain}
\bibliography{cycle} 
}

\onlyinproc{\end{document}}
\appendices

\section{Count distribution} \label{contdistc:sec}

  We provide the proof of Theorem \ref{cdensity:thm}.
\begin{proof}
 We first consider fixed-$\tau$ SH$_\ell$.
We use  $\Phi(w)$ \eqref{contcumphi}, which is the probability that a key of weight 
 $w$ is sampled.  This is the same as the probability of a key with
 $w_x \gg w$ starting to get counted after processing $y\leq w$ of its
 weight.
The partial derivative of $\Phi(w)$ with respect to 
 $w$ is the density function 
 $\phi(y)$ on the weight $y$  at which a key starts getting 
 counted:
\begin{eqnarray*}
\lefteqn{\phi(w) = \frac{\partial \Phi(w)}{\partial w}}\\
 &=&
\min\{1,\tau\ell\}\max\{1/\ell,\tau\} \exp(-w \max\{1/\ell,\tau\} \\
&=& \tau \exp(-w \max\{1/\ell,\tau\} \ . 
\end{eqnarray*}
For a particular key $x$, the density function of $c_x$ is equal to $\phi$ 
 when in the range $[0,w_x]$.
Elsewhere we have that
$\int_{w_x}^\infty \phi(y)dy=1-\Phi(w_x)$ is the probability 
that $x$ is not sampled.

  We now establish our claim for the fixed-size sampling algorithm. 
We start with a precise definition of the conditioning 
  we use.  The randomization used for a key $x$ includes 
the random hash value used in 
$\keybase{x}$, the randomization used to assign scores to all the elements of $x$, and the 
random $u_x,z_x$ (freshly drawn per eviction step) used to adjust the counters.   
Observe that given that
a key $x$ is cached, the threshold value $\tau$ only depends on the randomization 
of the other keys but not on that of $x$. 
When $x$ is not cached, the threshold value $\tau$ may depend 
on the randomization used for $x$.   But in this case,
$c_x=0$. 

We show that after each step, the distribution of the final value of
$c_x$ has the claimed density, when conditioned on the 
current threshold $\tau$.  
Correctness when $\tau$ does not change follows from the treatment of 
the fixed threshold case. 
It remains to consider eviction steps. 
Let $\tau$ be the threshold value before eviction and let $\tau^*$ be 
the new threshold value after eviction. 
If the new $\tau^*$ is determined by our key $x$, then our key $x$ is 
evicted. 
Otherwise, the new $\tau^*$ is determined by the $k$th largest $z_x$
among all other keys.  This value depends only on the randomization of 
other keys and not on $u_x,r_x$. 

  We now need to show that the particular computation we used for the 
  count adjustment preserves the claimed form of the distribution. For that we assume 
  that the distribution of $c_x$ was as claimed with respect 
  to the initial threshold $\tau$.  We express it as a function of the 
  final threshold $\tau^*$. 

  We first consider the case  $\tau\ell >
  1$ and $\tau^*\ell \geq 1$. 
With probability $\tau^*/\tau$ the density at $y$ is the same and 
with probability $(1-\tau^*/\tau)$, it is the integral over $u$ of the 
density with $\tau$ at $u<y$ and the density of  a new deduction,
which is $\Exp(\tau^*)$ at $y-u$. Now observe that 
the density of the deduction conditioned on $\tau$ before the adjustment was 
(by our assumption) $\tau e^{-y\tau}$.  We obtain 
{\small 
\begin{eqnarray*}
 \lefteqn{\frac{\tau^*}{\tau} \tau e^{-\tau y} +
 (1-\frac{\tau^*}{\tau})\int_0^y \tau e^{-\tau u}  \tau^* e^{-\tau^*  (y-u)} du}\\
 &=& \tau^*  e^{-\tau y} +\tau^*(\tau-\tau^*) e^{-\tau^* y}
 \int_0^y e^{-u(\tau-\tau^*)}du\\
 &=& \tau^* e^{-\tau y} + \tau^* e^{-\tau^* y} (1-e^{-y(\tau-\tau^*)})\\
 &=& \tau^* e^{-\tau^* y}
\end{eqnarray*}
}

We now consider the case that $\tau\ell > 1 $ and $\tau^*\ell < 1$. 
With probability $1-\tau^*\ell$, we have $\keybase{x} \geq \tau^*$ and 
the final count is $0$.  With probability 
$(\ell\tau)^{-1} \tau^*\ell= \tau^*/\tau$, we have $u_x \leq 1/(\ell\tau)$ and 
$\keybase{x} < \tau^*$ and the count remains the same.  Otherwise,
with probability $(1-(\ell\tau)^{-1})\tau^*\ell = \tau^*(\ell -
\tau^{-1})$ we have $u_x > 1/(\ell\tau)$ and $\keybase{x} < \tau^*$. 
In this case we consider the density $y$ of the sum of the previous $u$
and new deduction $(y-u) \sim \Exp(1/\ell)$.  We obtain 
{\small 
\begin{eqnarray*}
\lefteqn{\frac{\tau^*}{\tau} \tau e^{-\tau y} +
 \tau^*(\ell-\tau^{-1})\int_0^y \tau e^{-\tau u}
 (1/\ell) e^{-(y-u)/\ell}  du}\\
 &=& \tau^* e^{-\tau y} + \tau^*(\tau-\ell^{-1}) e^{-y/\ell}\int_0^y 
 e^{-(\tau-\ell^{-1})u}du \\
&=& \tau^* e^{-\tau y} +
\tau^* e^{-y/\ell}(1-e^{-(\tau-\ell^{-1})y}) \\
&=& \tau^* e^{-y/\ell}
\end{eqnarray*}
}

Last, we consider the case $\tau\ell < 1$. 
With probability $\tau^*/\tau$ the key maintains the same count, since 
conditioned on $\keybase{x}< \tau$, we have $\keybase{x}< \tau^*$ with probability 
$\tau^*/\tau$. 
Otherwise, the count is $0$. 
So we obtain the density 
$$\frac{\tau^*}{\tau}  \tau e^{-y/\ell}=\tau^* e^{-y/\ell}\ .$$
\end{proof}

\section{ppswor Variance Analysis}

 In our variance analysis, we make use of the following
notion of domination of a distribution by another distribution (or
  function):  A distribution on $y\geq 0$ with density function $b$ is
  {\em dominated} by
  a function $s$ if
$$ \forall z,\ \int_0^z b(y)dy \leq \int_0^z s(y)dy\ .$$
We will use domination to bound variance.  We will compute the
variance $v(y)$ conditioned on a threshold value $y$ and then
compute the unconditioned variance as the 
expectation of $v(y)$ over the distribution of the threshold.  We will
use a dominating distrbution to bount the variance using the
general property: 
\begin{lemma} \label{dom:lemma}
 If $v(y)$ is non-increasing in $y$ and $b$ is dominated by $s$
then
$\int_0^\infty b(y) f(y) dy \leq \int_0^\infty s(y) v(y) dy$.
\end{lemma}


  We are now ready to upper bound the CV for ppswor.  Our bounds for
  other sampling schemes  build on this ppswor proof.   We start with
  some basic lemma.
We use  the   notation $s_{W,k}$ and $S_{W,k}$ respectively for the density and 
  cummulative distribution functions of the Erlang distribution $\Erlang(W,k)$, which
  is a sum of $k$ independent
  exponential distribution with parameter $W$.
\begin{lemma} \label{Erlang:lemma}
Let $B$ be a  distribution which can be expressed as the sum of $k$ independent
exponential distributions, each with parameter that is at most $W$
(the set of parameters can be a random variable and parameters may not
be independent).  Then $B$ is dominated by $\Erlang(W,k)$.
\end{lemma}

  Consider now ppswor sampling of ${\cal X}$ with respect to weights
  $w_x$ and let $W=\sum_z w_z$.
For a particular key $x$, let $B_x$ be the distribution of the $k$th
smallest seed value in ${\cal
  X}\setminus x$.
\begin{lemma} \label{Bdomination:lemma}
For all keys $x$, the distribution $B_x$ is dominated by
$\Erlang(W,k)$.
\end{lemma}
\begin{proof}
Let $\tau' \sim B_x$.  By definition, $\tau'$ is the $k$th smallest  of independent
exponential random  variables with parameters $w_y$ for
 $y\in {\cal X}\setminus x$.
  From properties of the exponential
 distribution, the minimum seed is exponentially distributed with parameter
 $\sum_{y\in {\cal X}\setminus x} w_y  = W-w_x$. Conditioned on a
 particular key $z_1$ having the smallest seed,  the difference between the minimum and second smallest is exponentially
 distributed with parameter $W-w_x-w_1$, where $w_1$ is the weight of the
key $z_1$ with  minimum seed, and so on.   Therefore, the distribution on $\tau'$ conditioned on
 the ordered set of smallest-seed elements is a sum of $k$
 exponential random  variables with parameters {\em at most} $W$. The
 distribution  $B_x$ is a convex combination of such distributions:
 One distribution for each possible ordered subset of size $k$ of
 ${\cal X}\setminus x$, and each such choice has probability equal to
 the probability of the ordered subset being the first $k$ keys of
 ppswor sampling from ${\cal X}\setminus x$,.

Therefore, from Lemma~\ref{Erlang:lemma}, the distribution $B_x$ (for any $x$) is 
dominated  by Erlang with parameters $(W,k)$.
\end{proof}

\begin{theorem}  \label{ppsworcv:thm}
Consider ppswor sampling with respect to weights $f(w_x)$.
For  a segment $H$ with proportion $q = Q(f,H)/Q(f,{\cal X})$, the CV
of the estimate $\hat{Q}(f,H)$ \eqref{ppsworest} is at most $(q (k-1))^{-0.5}$.
\end{theorem}
\begin{proof}
  We extend the analysis in \cite{ECohen6f,ECohenADS:PODS2014} 
 (note that here we take $k$ to be the sample size without the
  threshold ($k+1$ smallest seed) whereas in
  \cite{ECohen6f,ECohenADS:PODS2014} $k$ is larger by $1$).

  WLOG, since we are considering sampling aggregated data, we assume $w_x\equiv f(w_x)$.
  Let $W = \sum_{x\in {\cal X}} w_x$ be the total weight of the   population.   

  We first consider the 
variance of the inverse probability estimate for a key with 
  weight $w$ with respect to a fixed threshold $\tau$.  The variance 
  is $(1/p -1)w^2$, where $p=1-e^{-w\tau}$ and is at most 
$$\var[\hat{w}_x \mid \tau] = w^2 \frac{e^{-\tau w}}{1-e^{-\tau w}}
\leq w/\tau\ ,$$
using the relation $e^{-x}/(1-e^{-x}) \leq 1/x$. 

We now consider the ``perspective'' of a key $x$ and the 
distribution $B_x$ of the $k$th smallest seed value in ${\cal
  X}\setminus x$.    From Lemma \ref{Bdomination:lemma}, $B_x$ is dominated by $\Erlang(W,k)$.

We will bound the variance of the estimate using the relation
$$\var[\hat{w}_x] = \E_{\tau' \sim B_x} \var[\hat{w}_x \mid \tau'] \ .$$

Since the conditioned variance $\var[\hat{w}_x \mid \tau']$ is non-increasing with
$\tau'$, from Lemma \ref{dom:lemma}, 
domination implies that 
$$\E_{\tau' \sim B_x} \var[\hat{w}_x  \mid \tau' ] \leq 
\E_{\tau' \sim S_{W,k}}\var[\hat{w}_x \mid \tau']\ .$$
Therefore it suffices to upper bound the expectation for $S_{W,k}$.

 We now use the Erlang density 
function \cite{Feller2} $$s_{W,k}(x) = \frac{W^k x^{k-1}}{(k-1)!}
e^{-Wx}\ $$ 
and the relation
$\int_0^\infty x^a e^{-bx} dx = a!/b^{a+1}$ to bound the variance:
\begin{eqnarray*}
\var[\hat{w_x}]  & \leq  &  \int_0^\infty s_{W,k}(z) \var[\hat{w}_x
\mid z] d z \\
& \leq &  \int_0^\infty \frac{W^{k} z^{k-2}}{(k-1)!} e^{-Wz}
         \frac{w}{z} dz \\
& \leq & w  \frac{W^{k}}{(k-1)!} \int_0^\infty z^{k-2} e^{-Wz}
         dz = \frac{wW}{k-1}\ .
\end{eqnarray*}

 Since  covariances between different keys are zero \cite{bottomk:VLDB2008}, 
the variance on a set $H$ with weight $w(H)$ is the
 sum of variances $\var[\hat{w}(H)] \leq  w(H) W/(k-1)$.  We divide by
 $w(H)^2$ and take the square root to obtain an upper bound on the CV.
\end{proof}

\section{SH sum statistics variance}
  We now consider the fixed-$k$ SH estimator applied for sum statistics
  $f(w)\equiv w$.  The estimate is $\hat{Q}(w,H)=\sum_{x\in H}
  (c_x+\tau^{-1})$ \cite{CCD:sigmetrics12}.  The estimator has at
  least the variance of the ppswor estimator, since it has the same
  distribution over keys, and the ppswor inverse probability estimate, which can
  not be applied with SH, minimizes variance for this distribution
  (over all unbiased estimators that can be expressed as a sum over sampled
  keys of per-key unbiased estimates).  Surprisingly, we obtain  the
  same upper bound on the CV:
\begin{theorem}  \label{aSHcv:thm}
The SH sum estimator on a segment $H$ with proportion $q =
w(H)/w({\cal X})$ has CV of at most $(q (k-1))^{-0.5}$.
\end{theorem}
\begin{proof}
We  bound the variance of the fixed-$k$ SH estimate on a key $x$
conditioned on $\tau$.  With probability $e^{-\tau w}$ the key is not
sampled, the estimate is $0$, and the contribution to variance is
$w^2$.  Otherwise, with density $\tau e^{-\tau y}$ the count is $c_x=w-y$, the
 estimate is $c_x+\tau^{-1}$, and the contribution to the variance is
 $(\tau^{-1}-y)^2$.
We obtain 
\begin{eqnarray*}
\var[\hat{w}_x \mid \tau] &=& w^2 e^{-w\tau} + \int_0^w \tau
                                e^{-y\tau} (\tau^{-1}-y)^2dy  \\
&=& \tau^{-2}(1-e^{-\tau w}) \leq w/\tau\ .
\end{eqnarray*}
The last inequality uses the relation $1-e^{-x} \leq x$.

  The distribution of $\tau'$,  the $k$th smallest seed in ${\cal X}\setminus x$,
  is the same as in ppswor and we can conclude as in the proof of Theorem \ref{ppsworcv:thm}.  We
take the expectation of this variance over the distribution
  $S_{W,k}$ which dominates  $B_x$
and using the zero covariances, bound 
the variance on $w(H)$ by $\frac{W w(H)}{k-1}$.
\end{proof}

\section{SH$_\ell$  cap statistics variance}

 We bound the variance by relating
the 
distribution of $\seed(x)$ under fixed-$k$ SH$_\ell$ to 
the distribution  with 
ppswor with respect to key weights $\min\{w_x,T\}$, that is, using 
$\seed{x} \sim \Exp[\min\{w_x,T\}]$.
  For the purpose of this analysis, we work with SH$_\ell$
  seed  distribution that is exponential when conditioned on $y<1/\ell$ instead of 
  uniform. This does not change the algorithm or estimation, since 
  there is a monotone transformation that preserves seed order.
The SH$_\ell$ density function of the seed of a key with weight $w$ is 
\begin{equation}  \label{shellseeddist}
b_{\ell,w}(y) = (y<1/\ell)\, ? \, \ell e^{-\ell y} \frac{1-e^{-w/\ell}}{1-1/e} \, : \,
w e^{-wy}\ .
\end{equation}

  We now can relate $b_{\ell,w_x}(y)$ to $\Exp[\min\{w_x,T\}]$:
\begin{lemma}  \label{seedrel:lemma}
For any key $x$ and cap value $T$, the density function
$b_{\ell,w_x}(y)$ of $\seed(x)$ under fixed-$k$ SH$_\ell$ is
dominated by 
$$\frac{e}{e-1} \max\{1,\ell/T\}   \min\{w_x,T\} e^{-\min\{w_x,T\} y}\
.$$ 
\end{lemma}
\begin{proof} 
We show that for any point $z\geq 0$,
{\small
\begin{eqnarray}
\lefteqn{\int_0^z b_{\ell,w}(y) dy}\nonumber \\
&\leq&  \int_0^z \frac{e}{e-1}
\max\{1,\frac{\ell}{T}\}   \min\{w,T\} e^{-\min\{w,T\} y}  dy \nonumber \\
&=& \frac{e}{e-1} \max\{1,\frac{\ell}{T}\} (1-e^{-\min\{w,T\} z})  \label{domby}
\end{eqnarray}
}


The proof is via case analysis.  We start with $z \leq 1/\ell$. We have that 
\begin{eqnarray*}
\int_0^z b_{\ell,w_x}(y) dy &=& \int_0^z \ell e^{-\ell y}
\frac{1-e^{-w/\ell}}{1-e^{-1}} \\
&=& (1-e^{-\ell z}) \frac{1-e^{-w/\ell}}{1-e^{-1}} 
\end{eqnarray*}
Therefore, to establish the claim \eqref{domby} we need to show that
\begin{equation} \label{domby2}
(1-e^{-\ell z}) (1-e^{-w/\ell}) \leq \max\{1,\frac{\ell}{T}\} (1-e^{-\min\{w,T\} z})
\end{equation}


\begin{itemize}
\item {\bf 
Case  $w<T$:}
$$(1-e^{-\ell z}) (1-e^{-w/\ell}) \leq (1-e^{-wz}) = (1-e^{-\min\{w,T\}
  z})$$ using the relation
\begin{equation} \label{exprelation}
\forall a,b\geq 0,\ (1-e^{-a})(1-e^{-b})\leq 1-e^{-ab}\ .
\end{equation}
and thus \eqref{domby2} holds.
\item
{\bf Case  $w\geq T$ and $\ell \geq T$:}
We have
\begin{eqnarray*}
(1-e^{- \ell z})(1-e^{-w/\ell }) &\leq& 1-e^{-\ell z} \leq
\frac{\ell}{T} (1-e^{-Tz})\\
&=& \frac{\ell}{T} (1-e^{-\min\{w,T\}z})
\end{eqnarray*}
The last inequality follows from the function
$(1-\exp(-x)/x$ being monotone decreasing:  Therefore $\ell \geq T$
implies $\ell z \geq Tz$ and thus
$\frac{(1-e^{-\ell z})}{\ell z} \leq \frac{(1-e^{-T z})}{T z}$ which
implies
$1-e^{-\ell z}) \leq \frac{\ell}{T}(1-e^{-T z})$.
We therefore obtain that \eqref{domby2} holds.
\item
{\bf Case  $w\geq T$ and $\ell\leq T$:} We have
\begin{eqnarray*}
(1-e^{-\ell z}) (1-e^{-w/\ell}) &\leq&  (1-e^{-T z}) \\
&=& (1-e^{-\min\{w,T\} z})\ .
\end{eqnarray*}
 and therefore \eqref{domby2} holds.
\end{itemize}

We now consider $z\geq 1/\ell$.
We have 
\begin{eqnarray*}
\int_0^{1/\ell} b_{\ell,w_x}(y) dy &=& 1-e^{-w/\ell} \\
\int_{1/\ell}^z b_{\ell,w_x}(y) dy &=& e^{-w/\ell} - e^{-w z}
\end{eqnarray*}
Thus,
$\int_{0}^z b_{\ell,w_x}(y) dy = 1-e^{-w z}\ .$
To verify \eqref{domby}, 
we need to show that for all $z\geq 1/\ell$,
\begin{equation} \label{domby3}
1-e^{-w z} \leq \frac{e}{e-1} \max\{1,\frac{\ell}{T}\}
(1-e^{-\min\{w,T\} z})\ .
\end{equation}
This is immediate for $w\leq T$.  We now consider $w\geq T$.
Since
$1-e^{-w z} \leq 1$, it suffices to show that
$$\frac{e}{e-1} \max\{1,\frac{\ell}{T}\} (1-e^{-T z}) \geq 1 \ .$$  
Using $z\geq 1/\ell$, it suffices to show
\begin{equation} \label{domby4}
\frac{e}{e-1} \max\{1,\frac{\ell}{T}\} (1-e^{-T/\ell}) \geq 1 \ .
\end{equation}
If $T\geq \ell$ then we substitute $1-e^{-{T/\ell}} \geq 1-e^{-1}=\frac{e-1}{e}$ to
show that \eqref{domby4} holds. 
If $T< \ell$ we have $T/\ell <1$ and use the inequality
$(1-e^{-a}) \geq a (1-e^{-1})$ for $a\leq 1$ to obtain:
$$\frac{e}{e-1} \max\{1,\frac{\ell}{T}\} (1-e^{-T/\ell}) \geq
\frac{e}{e-1} \frac{\ell}{T} \frac{T}{\ell} (1-e^{-1}) = 1\ .$$

\end{proof}


 We are now able to express, for a key $x$, a dominating distribution to the distribution of the $k$th
 smallest seed in ${\cal X}\setminus x$ when using fixed-$k$ SH$_\ell$.
\begin{lemma}  \label{dominateaSHell:lemma}
The distribution $B$ of the 
 $k$th smallest seed,  where seeds for $z\in {\cal X}\setminus x$ are
independently drawn from $b_{\ell,w_z}$, 
is dominated by the function
$$\frac{e}{e-1} \max\{1,\frac{\ell}{T}\} s_{W,k}\ ,$$
 where 
$$W=\sum_{z\in {\cal X}} \min\{w_z,T\}\ $$ and $s_{W,k}$ is the
density of $\Erlang(W,k)$.
\end{lemma}
\begin{proof}
From Lemma \ref{seedrel:lemma},  $B$ is dominated by 
$\frac{e}{e-1} \max\{1,\ell/T\}$  times the density of the $k$th seed
according to $\Exp[\min\{w_x,T\}]$.
The latter distribution is dominated by $s_{W,k}$.
We get the claim from transitivity 
of domination.
\end{proof}

\subsection{CV bound for 2-pass estimator} \label{2passcontashell:sec}
 We are now ready to bound the variance of the 2-pass fixed-$k$
 SH$_\ell$ estimator.  Recall that the
estimator is applied to an SH$_\ell$ sample which includes the exact
$\Cap_T(w_x)$ values of sampled keys.

 We first bound  the variance for a fixed $\tau$.
\begin{lemma} \label{twopassaSHellvar}
$$\var[\hat{\Cap}_T(w_x) \mid \tau ]  \leq 
\max\{\frac{T}{\ell},1\} \frac{\min\{w_x,T\}}{\tau}\ .$$
\end{lemma}
\begin{proof}
The inclusion probability of a key $x$ conditioned on
  the threshold $\tau$ is
\begin{equation}\label{probtau:eq}
\Pr[\seed{x}< \tau] = \bigg\{ \begin{array}{lll} \tau\ell >1  & \text{:} & (1-e^{-\tau w_x}) \\
\tau\ell \leq 1 & \text{:} & (1- e^{-\ell \tau}) \frac{1-e^{-w_x/\ell}}{1-1/e}\end{array}\bigg. 
\end{equation}

The variance conditioned on $\tau$ of the inverse
  probability estimate is
\begin{equation} \label{vartauShell}
\var[\hat{\Cap}_T(w_x) \mid \tau ] = (\frac{1}{\Pr[\seed{x}<\tau]}-1)
\min\{w_x,T\}^2 \ .
\end{equation}

 For $\tau \ell > 1$ we have 
\begin{eqnarray*}
(\frac{1}{\Pr[\seed(x)<\tau]}-1) &\leq&
 \frac{e^{-\tau w_x}}{1-e^{-\tau w_x}} \\
&=& \frac{1}{e^{\tau w_x}-1} \leq \frac{1}{\tau  w_x}\ .
\end{eqnarray*}
The last inequality uses the relation $e^x\leq 1+x$ for $x\geq 0$.
Substituting in \eqref{vartauShell} we obtain 
$$\var[\hat{\Cap}_T(w_x) \mid \tau ]  \leq
\min\{w_x,T\}^2  \frac{1}{\tau  w_x} \leq \frac{\min\{w_x,T\}}{\tau}\
.$$

It remains to treat the case $\tau\ell \leq 1$.
We will establish that
\begin{equation}  \label{midclaim}
\frac{1-1/e}{(1-e^{-\ell \tau})(1-e^{-w/\ell})} -1 \leq \frac{1}{\min\{w_x,\ell\}\tau}\ .
\end{equation}

 We first consider $w \geq \ell$.  In this case the right hand size is
 fixed at $1/(\ell\tau)$.  To maximize the left hand size, 
over $w\geq \ell$, we
 take $w=\ell$.  We then obtain that the left hand size is at most
$\frac{e^{-\ell\tau}}{1-e^{-\ell\tau}} \leq \frac{1}{\ell \tau}$, which
establishes \eqref{midclaim}.

  We next consider $w\leq \ell$, recalling that we already assume $\tau
  \ell<1$, and thus have $w<\ell<1/\tau$.
To maximize the left hand side of \eqref{midclaim} under these
assumptions we need to minimize the denominator
$h(\ell)\equiv (1-e^{-\ell \tau})(1-e^{-w/\ell})$ in the range
$w<\ell<1/\tau$.
By taking the derivative 
$\frac{\partial h}{\partial \ell} \geq 0$, we see that it is negative
in this range.   Therefore, $h(\ell)$ is
minimized for $\ell=1/\tau$.
Substituting $\ell=1/\tau$, we obtain that the left hand size of
\eqref{midclaim} is at most
$\frac{e^{-w\tau}}{1-e^{-w\tau}} \leq \frac{1}{w\tau}$, and thus
\eqref{midclaim} is fully established.

 We now note that the left hand size of \eqref{midclaim} is equal to
 $\frac{1}{\Pr[\seed(x)<\tau]}-1$.
Substituting in \eqref{vartauShell}, we obtain that
\begin{eqnarray*}
\lefteqn{\var[\hat{\Cap}_T(w) \mid \tau ]} \\ 
& = & \bigg(\frac{1}{\Pr[\seed{x}<\tau]}-1\bigg)
\min\{w,T\}^2 \\ 
&\leq&  \frac{1}{\ell \tau}\min\{w,T\}^2 \\
&\leq&
\min\{\frac{w}{\ell},\frac{T}{\ell}\} \frac{\min\{w,T\}}{\tau} \\
&\leq&
\min\{1,\frac{T}{\ell}\} \frac{\min\{w,T\}}{\tau}
\end{eqnarray*}
The second to last inequality uses our assumption that $w\leq \ell$.
\end{proof}

We are now ready to conclude the proof of Theorem \ref{cv2passSHell:thm}.
  We bound the variance with respect to the 
  distribution $B$ using the dominating function as in Lemma \ref{dominateaSHell:lemma}. 
 We
  obtain that the variance for key $x$ is
$$\var[\hat{\Cap}_T(w_x)] \leq \frac{e}{e-1}  \max\{\frac{T}{\ell},\frac{\ell}{T}\}
 W \frac{\min\{w_x,T\}}{k-1}\ .$$  We conclude as in the proof of
 Theorem \ref{ppsworcv:thm}, showing that our estimate of $Q(\Cap_T,H)$
for a segment $H$ with
 proportion $q$ of the $\Cap_T$ statistics  has CV that is at most
$\bigg(\frac{e}{e-1} \frac{\max\{T/\ell,\ell/T\}}{q(k-1)}\bigg)^{0.5}$.

\subsection{1-pass variance bound}  \label{1passcontashell:sec}

We provide the proof of Theorem~\ref{cv1passSHell:thm}, which bounds
the variance of the 1-pass estimators.  

The estimators are applied to the same sample distribution of included 
keys and the proof outline is similar to that of the 2-pass estimator
(Theorem~\ref{cv2passSHell:thm}).  The only component we are missing
is
a bound on the conditional variance $\var[\hat{\Cap}_T(w_x) \mid \tau ]$:
  Since the exact weight $w_x$ is not available,  we can not
apply Lemma \ref{twopassaSHellvar} and instead compute the variance of
the 1-pass estimator that is applied to  $c_x$.

\begin{lemma}
\begin{eqnarray*}
\var[\hat{\Cap}_T(w_x) \mid \tau ]  & \leq &
\frac{\min\{w_x,T\}}{\tau} \bigg(\frac{\ell}{T}(1-e^{-T/\ell}) +
\frac{T}{\ell}\bigg)\\
&\leq& (1+\frac{T}{\ell})
\frac{\min\{w_x,T\}}{\tau}\ .
\end{eqnarray*}
\end{lemma}
\begin{proof}

The estimation coefficients $\beta(c)$ are provided in Theorem \ref{estcoeffcont:thm}.
We bound
\begin{eqnarray*}
\var[\hat{\Cap}_T(w) \mid \tau ] & =&  \E[\beta(c)^2]-\E[\beta(c)]^2
\\
& = & \E[\beta(c)^2]-\min\{T,w\}^2\ .
\end{eqnarray*}
The last inequality follows from unbiasedness
$\E[\beta(c)] = \Cap_T(w)$.

 For $\Cap_T$, we have $f'(x) = 1$ for $x\leq T$ and $f'(x)=0$
 otherwise.  Therefore,
$$\beta(c) = \frac{\min\{T,c\}}{\min\{1,\ell\tau\}} + \tau^{-1}
I_{c<T} \ .$$

We use the density function of $c$ given $w$ which is provided in Theorem
\ref{cdensity:thm}.
The density function for $w \geq c>0$ is $\tau e^{-\tau(w-c)}$ when
$\tau\ell \geq 1$  and $\tau e^{-(w-c)/\ell}$ when $\tau\ell \leq 1$.
We have that $\beta(c)=0$ when $c=0$.
We now use case analysis.

We first consider the case $\tau\ell >1$ and $w\leq T$.
We have $\beta(c)= c +1/\tau$.
\begin{eqnarray*}
\E[\beta^2]-w^2  & = & -w^2 + \int_0^w \tau e^{-\tau x} \beta(w-x)^2 dx \\
&=& -w^2 + \int_0^w \tau e^{-\tau x} (\tau^{-1}+w-x)^2 dx\\
&=& (1- e^{-\tau w})(w^2+\tau^{-2})-w^2 \leq w/\tau 
\end{eqnarray*}
The last inequality uses the relation $(1-e^{-x}) \leq x$.

For $\tau\ell >1$ and $w> T$, we have $\beta(c)=T$ for $w\geq c\geq T$
and $\beta(c)= c+1/\tau$ for $c<T$. Therefore,
{\small
\begin{eqnarray*}
\lefteqn{\E[\beta^2] - T^2 }\\
& = & -T^2+ \int_0^{w-T} \tau e^{-\tau x} T^2 dx +
\int_{w-T}^w \tau e^{-\tau x} (\tau^{-1}+w-x)^2 dx \\
&=& -T^2 + (1-e^{-\tau (w-T)})T^2  + e^{-\tau (w-T)}(T^2+\tau^{-2}) -\tau^{-2}e^{-\tau w})\\
&=&  \tau^{-2} e^{-\tau w}(e^{\tau T}-1) \leq \tau^{-2}(1-e^{-\tau T})
    \leq T/\tau\ .
\end{eqnarray*}
}
Using the fact that $e^{-Tw}$ is maximized (subject to $w\geq T$) when
$w=T$ and $(1-e^{-x}) \leq x$.

For $\tau\ell <1$ and $w\leq T$, we have $\beta(c)=\tau^{-1}(c/\ell+1)$.
\begin{eqnarray*}
\E[\beta^2] & = & \int_0^w \tau e^{- x/\ell}
\tau^{-2}(\frac{w+\ell-x}{\ell})^2 dx \\
& = & \frac{1}{\tau\ell^2} \int_0^w e^{- x/\ell} (w+\ell-x)^2 dx \\
&=& \frac{\ell}{\tau}(1-e^{-w/\ell}) + \frac{w^2}{\tau\ell} \\
&=& \frac{w}{\tau}\bigg(\frac{\ell}{w}(1-e^{-w/\ell}) +
\frac{w}{\ell}\bigg)\\
&\leq& \frac{w}{\tau}\bigg(\frac{\ell}{T}(1-e^{-T/\ell}) +
\frac{T}{\ell}\bigg)\\
&\leq&
    \frac{w}{\tau}(1+\frac{T}{\ell})\ .
\end{eqnarray*}
The second to last inequality follows from the function
$(1-e^{-x})/x+x$ being increasing in the range $x>0$.  Therefore, 
subject to fixed $\ell$ and $w\leq T$, the expression using $x= w/\ell$ is maximized at $w=T$.

For $\tau\ell < 1$ and $w> T$:
\begin{eqnarray*}
\E[\beta^2] & = & \int_0^{w-T} \tau e^{- x/\ell} T^2/(\tau\ell)^2 dx
                  +\\
&&                   \int_{w-T}^w \tau e^{- x/\ell} (\tau^{-1}+(w-x)/(\tau\ell))^2 dx \\
&=& (\tau\ell)^{-1} T^2 (1-e^{-(w-T)/\ell}) + \\ &&  (\tau\ell)^{-1}
    \int_{w-T}^w (1/\ell)e^{- x/\ell} (\ell+w-x)^2dx\\
&=& (\tau\ell)^{-1}\bigg( T^2 (1-e^{-(w-T)/\ell}) + \\
  &&    e^{-(w-T)/\ell}(\ell^2+T^2)-\ell^2 e^{-w/\ell} \bigg)\\
&=& (\tau\ell)^{-1}\bigg( T^2 +  e^{-(w-T)/\ell} \ell^2 -\ell^2
    e^{-w/\ell} \bigg)\\
&=& \frac{T}{\tau}\bigg( \frac{T}{\ell} +
\frac{\ell}{T}e^{-w/\ell}(e^{T/\ell} -1)\bigg) \\
&\leq& \frac{T}{\tau}\bigg( \frac{T}{\ell} +
\frac{\ell}{T}(1-e^{-T/\ell}) \bigg) \\
&\leq & \frac{T}{\tau}(1+T/\ell) 
\end{eqnarray*}
The second to last derivation substitutes $w=T$ for the $w$ value that
maximizes the expression subject to $w\geq T$.  We then use
$(1-e^{-x})\leq x$.
\end{proof}

\section{Discretized threshold sampling}

  A variation on the fixed-$k$ discrete SH$_\ell$ sampling scheme
that can be useful in practice
is to limit the algorithm to work with a discrete set of thresholds (think $\tau=\alpha ^i$ for 
  some $\alpha<1$) (see Algorithm~\ref{dtellaSH:alg}).  When 
  the cache is full, the threshold is adjusted in iteration until its 
  size drops below $k$.  
Discretized thresholds with fixed-$k$ SH were 
  considered in \cite{flowsketch:JCSS2014}. 
Discretized thresholds have the advantage that the number of times 
keys are pulled out/placed back on the priority queue for updates is lower. 
Another advantage is that the 
  estimators, when expressed as coefficients which depend on 
  $\ell$ and $\tau$,  can be reused with different samples. 

\ignore{
\begin{procedure}[h]{\caption{AdjustCountScore($\tau,x$)}}
\DontPrintSemicolon 
\KwIn{key $x$ currently in $\Counters$, fixed $\ell$, threshold $\tau$}
\KwOut{Decrease count until $0$ and key is removed from $\Counters$ or until $\seed{x}<\tau$}
\SetKwFunction{rand}{rand}
\SetKwFunction{Hash}{Hash}
\SetKwFunction{Return}{return}
\While{$\Counters{y} > 0$ and $\seed{y} \geq \tau$}
{$\Counters{y} \gets \Counters{y}-1$\;
  $\seed{y} \gets \ElementScore{y}$\;      } 
\If{$\Counters{y} == 0$}{ delete $\Counters{y}$}\;
\end{procedure}
}

\begin{algorithm2e}[h]
\caption{stream sampling: max size $k$ and discretized 
  thresholds \label{dtellaSH:alg}}
\DontPrintSemicolon 
\SetKwArray{Counters}{Counters}
\SetKwFunction{rand}{rand}
\SetKwFunction{Hash}{Hash}
\KwData{sample size $k$, $\alpha<1$, stream of elements from ${\cal X}$}
\KwOut{set of $k$ pairs $(x,c_x)$ where $x\in {\cal X}$ and $c_x \in [1,w_x]$}
\SetKwFunction{Return}{return}
$\Counters \gets \emptyset$  \tcp*[h]{Initialization} \;
$\tau \gets 1$ \tcp*[h]{Sampling Threshold} \;
\tcp{Processing a stream element of key $x$}
\ForEach{stream element $h$ with key $x$}
{
  \If{$x$ is in $\Counters$}{$\Counters{x} \gets \Counters{x}+1$\;}
  \Else{
    $\seed{x} \gets \ElementScore{h}$\;
    \If{$\seed{x} < \tau$}
    {
      $\Counters{x} \gets 1$\;
      \While {$|\Counters| > k$}
      {
        $\tau \gets \alpha \tau$\;
        \While {$\max\{\seed{x} \mid x \text{ in } \Counters\} \geq 
          \tau$}
        {
          $y \gets \arg\max\{\seed{x} \mid x \text{ in } \Counters\}$\;
         \While{$\Counters{y} > 0$ and $\seed{y} \geq \tau$}
         {$\Counters{y} \gets \Counters{y}-1$\;
           $\seed{y} \gets \ElementScore{y}$\;      } 
         \If{$\Counters{y} == 0$}{ delete $\Counters{y}$, $\seed{y}$}\;
        }
      }
    }
  }
}
\Return{$\tau$ ;  $(x,\Counters{x})$ for $x$ in $\Counters$}
\end{algorithm2e}

\section{Approximate $\Cap_T$ Counters}

  Our design computes a sample of the active keys, over which 
segment statistics can be estimated.
One can also consider the more basic problem of only estimating the statistics 
  over the full data set $Q(f,{\cal X})$. 

The case $\Cap_1$ corresponds to a (approximate) distinct count of 
  keys, which is a fundamental and well-studied problem.
The case $\Cap_\infty$ is the total weight of
  the full stream and can easily be computed. 
 
 Our constructions can be modified to be more efficient when we are
 only interested in estimating $Q(\Cap_T,{\cal X})$:  Instead of 
storing full identifiers
  of cached keys, which are needed for a sample,  we can hash the key domain to a domain of size 
  that is polynomial in the number $n$ of distinct keys.  The resulting sketch size 
in this case would be $O(\log n)$ to represent each key hash and 
the count (which we can cap by a polynomial in $T$, to ensure the
representation of the counts is at most $O(\log T)$.   The result is an
 approximate $\Cap_T$ counter on streams that has state
(structure size) that is $O(\epsilon^{-2} (\log T  + \log n)$ and
provides estimates with  CV  of $\epsilon$.  

  State of the art approximate distinct counters, however, have a
  smaller, double logarithmic dependence on $n$
  \cite{hyperloglog:2007,hyperloglogpractice:EDBT2013}.
We present here a light weight algorithm that provides a rough
approximation of $Q(\Cap_T,{\cal X})$ with double logarithmic state.
We apply to each stream element the string returned by the element scoring function
$\ElementScore{h}$ \eqref{delementscore:eq}  used
in our discrete SH$_T$ algorithm.  We then apply any off-the-shelf
approximate distinct counter 
  \cite{FlajoletMartin85,hyperloglog:2007,hyperloglogpractice:EDBT2013,ECohenADS:PODS2014}
  to the
stream of  $\ElementScore{h}$.
Recall that the  elements being counted are the
identifiers of key-bucket pairs from the original stream, where a bucket $b \sim U[1,\ldots,T]$
    is drawn independently for each stream appearance of the key. 

  We now analyse the quality of this approximation.
\begin{lemma}
The expected number of distinct strings generated is between
$(1-1/e) Q(\Cap_T,{\cal X})$ and $Q(\Cap_T,{\cal X})$\ .
\end{lemma}
\begin{proof}
The expected number of distinct strings  that are 
  generated for a key of cardinality $w$ is 
  $\ell(1-(1-1/T)^w)$.  This is because the probability that we do 
  not hit a certain bucket with $w$ elements is $(1-1/T)^w$.  Thus,
  the expected number of empty buckets is $T(1-1/T)^w$. 

  So in expectation, a distinct counter applied 
with $T$ buckets would produce an underestimate.   The worst relative 
  error is obtained for keys with $w=T$, where the expected count is 
  $T(1-1/e)$, thus the relative error is $1/e$.   However, the 
  error depends on the distribution of key sizes, and is small for 
  cardinalities much larger or smaller than $T$.  
\end{proof}
This approach can obtain a rough estimate of $Q(\Cap_T,{\cal X})$
 to within $(1-1/e,1)$ using state of size $O(\log\log n)$.  (Since
 there is inherent error of $1-1/e$, there is no point in using a more
 accurate distinct counter)

  One approach to reduce the error, left for future work, is to
apply the counting with multiple  values of $T$.
\end{document}